\documentclass[sigconf]{acmart}
\usepackage{booktabs} 
\usepackage{color}
\usepackage{soul}
\usepackage{graphicx}

\usepackage{subfigure}
\usepackage{mathrsfs}
\usepackage{algorithmic}
\usepackage[toc,page]{appendix}
\usepackage{amsmath,epsfig}
\usepackage{mathrsfs}
\usepackage{amssymb}

\usepackage{etoolbox}
\usepackage{xcolor}
\usepackage{balance}

\usepackage[linesnumbered,ruled,vlined]{algorithm2e}

\SetCommentSty{mycommfont}

\SetKwRepeat{Do}{do}{while}

\DeclareMathOperator*{\argmin}{arg\,min}
\setcopyright{none}

\begin{document}
\title{VideoDP: A Universal Platform for Video Analytics with Differential Privacy}

\author{Han Wang}

\affiliation{
  \institution{Illinois Institute of Technology}
}
\email{hwang185@hawk.iit.edu}

\author{Shangyu Xie}

\affiliation{
  \institution{Illinois Institute of Technology}
}
\email{sxie14@hawk.iit.edu}

\author{Yuan Hong}

\affiliation{
  \institution{Illinois Institute of Technology}
}
\email{yuan.hong@iit.edu}

\begin{abstract}
Massive amounts of video data are ubiquitously generated in personal devices and dedicated video recording facilities. Analyzing such data would be extremely beneficial in real world (e.g., urban traffic analysis, pedestrian behavior analysis, video surveillance). However, videos contain considerable sensitive information, such as human faces, identities and activities.  
Most of the existing video sanitization techniques simply obfuscate the video by detecting and blurring the region of interests (e.g., faces, vehicle plates, locations and timestamps) \emph{without quantifying} and \emph{bounding the privacy leakage} in the sanitization. 
In this paper, to the best of our knowledge, we propose the first differentially private video analytics platform (\emph{VideoDP}) which flexibly supports different video analyses with rigorous privacy guarantee. Different from traditional noise-injection based differentially private mechanisms, given the input video, \emph{VideoDP} randomly generates a utility-driven private video in which adding or removing any sensitive visual element (e.g., human, object) does not significantly affect the output video. Then, different video analyses requested by untrusted video analysts can be flexibly performed over the utility-driven video while ensuring differential privacy. 
Finally, we conduct experiments on real videos, and the experimental results demonstrate that our \emph{VideoDP} effectively functions video analytics with good utility.
\end{abstract}

\maketitle

\section{Introduction}
\label{sec:intro}
Massive amounts of video data are ubiquitously generated everyday from many different sources such as personal cameras and smart phones, traffic monitoring and video surveillance facilities, and many other video recording devices. Analyzing such complex, unstructured and voluminous data \cite{videodata} would be extremely beneficial in real world (e.g., video surveillance).
For instance, traffic monitoring videos can be analyzed by traffic authorities, urban planning officials, and some researchers \cite{Abreu:MultiAgent} for learning urban traffic and pedestrian behavior. Videos recorded by surveillance devices (generally involving numerous persons) might be analyzed for detecting anomalies or suspicious behavior.

However, directly releasing videos to the analysts would result in severe privacy concerns due to the sensitive information involved in videos, such as human faces, objects, identities and activities \cite{VPPAAct}. For instance, traffic monitoring cameras can capture all the vehicles which may involve the make, model and color of vehicles, moving speed and trajectories, and even the drivers' faces. To this end, video sanitization (e.g., traffic monitoring videos \cite{Abreu:MultiAgent}, surveillance videos \cite{SainiMTAP_W3}, YouTube videos \cite{youtube}) have been recently studied.  
Most of such techniques (including the YouTube Blurring application \cite{youtube}) obfuscate the video by \emph{detecting} and then directly \emph{blurring} the region of interests, e.g., faces, vehicle plates, and locations \cite{SainiMTAP_W3,Hill2016}. 
However, \textit{privacy leakage} in the sanitized videos cannot be bounded with a rigorous privacy notion. Then, video owners or individuals cannot be provided with quantified privacy risks. Without rigorous privacy guarantee, the blurred regions might be reconstructed by deep learning methods \cite{dl1,dl2}.

To address such deficiency, we propose a novel platform (namely, \emph{VideoDP}) that ensures differential privacy  \cite{Dwork06} for any video analysis requested from untrusted data analysts, including queries or query-based analyses over the input video. Notice that, as the state-of-the-art privacy model, differential privacy (DP) \cite{Dwork06} can ensure indistinguishable analysis result for inputs with and without any one record (protecting any record against arbitrary background knowledge). In \emph{VideoDP}, we define a novel differential privacy notion that adding or removing any sensitive visual element (e.g., human, object) into the input video does not significantly affect the analysis result. Thus, the privacy risks can be strictly bounded even if the adversaries possess arbitrary background knowledge (e.g., knowing the objects or humans in the input video). To the best of our knowledge, this is the first work proposed to provide differentially private video analysis. Specifically, in \emph{VideoDP}, we address the following unique challenges (different from the existing differentially private schemes applied to other datasets, e.g., \cite{Dwork06,KorolovaKMN09,Cormode12,HayLMJ09,ninghui13,zhan17}). Specifically,

\vspace{0.05in}

\noindent\textbf{Differential Privacy.} Most other datasets (e.g., statistical data \cite{Dwork06}, location data \cite{ninghui13}, search logs \cite{KorolovaKMN09}, social graphs \cite{zhan17}) have explicit attribution of privacy concerns w.r.t. different individuals. The differential privacy is defined as ``adding or removing any individual's data does not result in significant privacy leakage''. 
In the context of videos, as mentioned earlier, we consider the ``appearance of sensitive objects or humans'' as the root cause of privacy leakage in videos, and then seek for the protection that the untrusted analyst cannot \emph{distinguish if any sensitive object or human is included in the video or not}, even if the adversaries have arbitrary background knowledge about the objects/humans. Then, we first address the challenge on accurately detecting and tagging all the objects/humans in any video (by utilizing state-of-the-art computer vision techniques  \cite{tracking1,tracking2}). For instance, given a video recorded on the street, our objective is to protect sensitive objects (e.g., vehicles) and/or humans (e.g., pedestrians) by the DP scheme. 

\vspace{0.05in}

\noindent\textbf{Utility-driven Private Video}. 
Given any input video, different from traditional differentially private mechanisms (e.g., injecting noise into queries/analyses), we propose a novel randomization scheme (via pixel sampling) to generate a utility-driven private video while ensuring the defined differential privacy notion. 
Specifically, our \emph{VideoDP} involves three phases. The first phase (\emph{differentially private pixel sampling}) randomly generates pixels for the output video based on the visual elements and background scene in the input video. 
Since videos are extremely large scale and highly-dimensional at the pixel level (generally consisting of millions of pixels with very diverse RGBs \cite{Acharya:2005:IPP:1088917}), it is extremely challenging to ensure good utility for video via pixel sampling (e.g., many RGBs/pixels cannot be sampled). 

To further improve the output utility, after executing pixel sampling in \emph{VideoDP}, the second phase generates the (random) \emph{utility-driven private video} by interpolating the RGB values of unsampled pixels and integrates such ``estimated pixels'' into the missing pixels. Note that the pixel interpolation also satisfies differential privacy with \emph{indistinguishability}.

\vspace{0.05in}

\noindent\textbf{Universal Video Analytics Platform}.
In the first two phases, our \emph{VideoDP} generates the (probabilistic) utility-driven private video which ensures \emph{indistinguishability} regardless of adding or removing any sensitive visual element in the input video. Therefore, in the third phase, different video analyses requested by untrusted data analysts (e.g., queries over the video for analytics) can be flexibly performed over the utility-driven private video, as analyzed in Section \ref{sec:pa}. \emph{VideoDP} significantly outperforms the PINQ platform \cite{McSherry09} in the context of video analytics (e.g., reduced perturbation, flexiblity for different video analyses), as discussed in Section \ref{sec:frame} and validated in the experiments (Section \ref{sec:exp}).

\vspace{0.05in}

\noindent\textbf{Contributions.} Motivated by first exploring rigorous privacy guarantee in videos, the major contributions of this paper are summarized as below:

\begin{itemize}

\item We define the first differential privacy notion with respect to protecting all the sensitive visual elements in any video. 

\item We propose a novel platform \emph{VideoDP} which can flexibly perform any video analysis requested by the video analyst with differential privacy guarantee.  

\item  \emph{VideoDP} randomly generates a utility-driven private video by sampling pixels (Phase I) with differential privacy and interpolating unsampled pixels (Phase II) to boost the utility for video analytics. Then, it enables universal private video analyses (Phase III) with untrusted analysts.

\item We have conducted extensive experiments to validate the performance of \emph{VideoDP} on real videos. 
\end{itemize}

The remainder of this paper is organized as follows.
Section \ref{sec:model} introduces some preliminaries. Section \ref{sec:algm} illustrates the first phase of \emph{VideoDP} and analyzes the privacy guarantee. Section \ref{sec:post} presents the second phase and third phase (private video analytics) as well as the differential privacy guarantee. Section \ref{sec:disc} discusses some relevant aspects of \emph{VideoDP}. Section \ref{sec:exp} demonstrates the experimental results. Section \ref{sec:related} and \ref{sec:concl} present the literature and conclusions.

\section{Preliminaries}
\label{sec:model}
In this section, we present some preliminaries for our \emph{VideoDP}. Since the differential privacy notion defined in \emph{VideoDP} protects all the sensitive objects and humans in the videos, for simplicity of notations, we use ``VE'' to represent both objects and humans (viz. visual elements) such as vehicles and pedestrians in this paper. 

\subsection{Video Features and Notations}

Any video streamlines a sequence of frames, each of which includes a fixed number of pixels. Referring to the RGB color model \cite{Acharya:2005:IPP:1088917}, video data includes \emph{frame ID, 2-D coordinates, red, green, blue}.\footnotemark[1] Specifically, we denote any pixel's frame ID as $t$, its coordinates as ($a$,$b$), and its RGB as a 3-dimensional vector $\theta(a,b,t)\in [0,255]^3$ (\emph{16,581,375 distinct RGBs in the universe}). 
\footnotetext[1]{We focus on the privacy leakage resulted from the visual elements rather than audio.}

\vspace{0.05in}

\noindent\textbf{VE Detection.} The state-of-the-art computer vision algorithms can be utilized to accurately detect VEs (e.g., objects \cite{Girshick15} and humans \cite{humandetection}) in videos. In \emph{VideoDP}, all the VEs in a video (denoted as $\Upsilon_j, j\in[1,n]$) are detected using the tracking algorithm \cite{tracking1,tracking2} in which the same human/object in different frames is assigned the same unique identifier (see Section \ref{sec:exp} for details). Since the computer vision algorithms \cite{tracking1,tracking2} achieve a very high detection accuracy, we assume that the undetected visual elements do not result in privacy leakage w.r.t. individuals in this paper.

Notice that the same visual element $\Upsilon_j$ may have different size and different RGB values (e.g., as a vehicle moves close to the camera, its size visually grows). Thus, we should investigate all the RGBs of every VE involved in all the frames. To break down the video into pixels with RGBs, we denote the set of distinct RGBs in VE $\Upsilon_j$ (in all the frames) as $\Psi_j$ where the cardinality is written as $|\Psi_j|$ (the number of distinct RGBs in $\Upsilon_j$). 
Table \ref{table:notation} in Appendix \ref{sec:notation} shows the notations for video features in this paper.

\begin{figure*}[!tbh]
	\centering
		{\includegraphics[angle=0, width=0.85\linewidth]{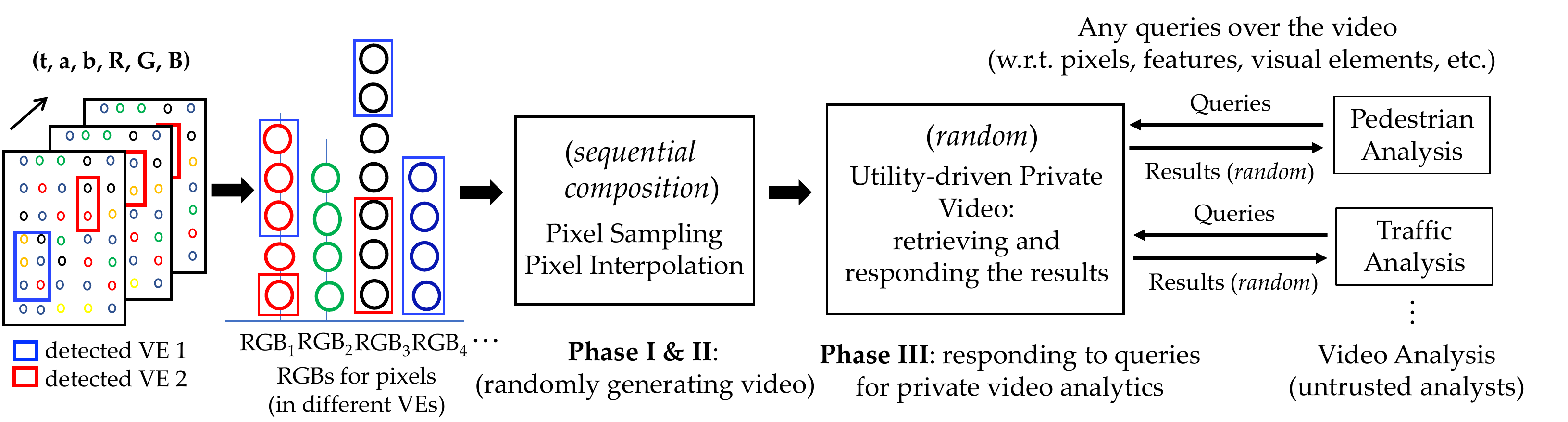}}
	\caption[Optional caption for list of figures]
	{\emph{VideoDP} Framework ($\epsilon$-differential privacy for Phase I--III)}\vspace{-0.1in}
	\label{fig:overview}
\end{figure*}

\subsection{Privacy Model}
\label{sec:privacy}
To protect sensitive objects or humans (``VE'') in the video, we first consider two input videos $V$ and $V'$ that differ in any visual element $\Upsilon$ (in all the frames) as two neighboring inputs.  
Specifically, given a video $V$, after completely removing $\Upsilon$ in all the frames of $V$, we can obtain $V'$ (or vice-versa). Note that $V$ and $V'$ have identical number of frames and background scene.

Then, \emph{VideoDP} ensures that adding any VE into any number of frames in a video or completely removing any VE from the video would not result in significant privacy risks in video analytics, assuming that the adversary possesses arbitrary background knowledge on all the VEs. W.l.o.g., denoting $V=V'\cup \Upsilon$, we have:

\begin{definition}[$\epsilon$-Differential Privacy] A randomization algorithm $\mathcal{A}$ satisfies $\epsilon$-differential privacy if for any two input videos $V$ and $V'$ that differ in any visual element (e.g., object or human) $\Upsilon$, and for any
output $O\in range(\mathcal{A})$, we have
$e^{-\epsilon}\leq \frac{Pr[\mathcal{A}(V)=O]}{Pr[\mathcal{A}(V')=O]}\leq e^{\epsilon}$.
\label{def:dp}
\end{definition}

Furthermore, in two neighboring videos $V$ and $V'$, if there exists a possible output $O\in range(\mathcal{A})$ that makes any of $Pr[\mathcal{A}(V)=O]$ and $Pr[\mathcal{A}(V')=O]$ equal to 0. For instance, since the extra visual element $\Upsilon$ is included in $V$ but not in $V'$, if an output $O$ involves elements from $\Upsilon$ but cannot be generated from $V'$ (simply due to $\Upsilon\cap V'=\emptyset$). At this time, for such output $O$, we have $Pr[\mathcal{A}(V)=O]>0$ while $Pr[\mathcal{A}(V')=O]=0$. 

In such cases, the multiplicative difference between $\frac{Pr[\mathcal{A}(V)=O]}{Pr[\mathcal{A}(V')=O]}$ and $\frac{Pr[\mathcal{A}(V')=O]}{Pr[\mathcal{A}(V)=O]}$ cannot be bounded by $e^\epsilon$ (due to the zero denominator). To accommodate this, a relaxed privacy notion \cite{MachanavajjhalaKAGV08,corn09} can be defined:

 \begin{definition}[$(\epsilon,\delta)$-Differential Privacy \cite{MachanavajjhalaKAGV08,corn09}] A randomization algorithm $\mathcal{A}$ satisfies $(\epsilon,\delta)$-differential privacy if for all video $V$, we can divide the output space $range(\mathcal{A})$ into two sets $\Omega_1,\Omega_2$ such that (1) $Pr[\mathcal{A}(V)\in\Omega_1]\leq \delta$, and (2) for any of $V$'s neighboring video $V'$ and for all $O\in \Omega_2$: (2) $e^{-\epsilon}\leq \frac{Pr[\mathcal{A}(V)=O]}{Pr[\mathcal{A}(V')=O]}\leq e^{\epsilon}$. 

\label{def:rdp}
\end{definition}

This definition guarantees that algorithm $\mathcal{A}$ achieves $\epsilon$-differential privacy with a high probability $(\geq 1-\delta)$ \cite{MachanavajjhalaKAGV08,corn09}. The probability that generating the output with unbounded multiplicative difference for $V$ and $V'$ is bounded by $\delta$.

\subsection{\emph{VideoDP} Framework}
\label{sec:frame}

\subsubsection{Limitation of PINQ-based Video Analytics}
Privacy Integrated Queries (PINQ) \cite{McSherry09} platform was proposed to facilitate data analytics by injecting Laplace noise into the queries required by the analysis. Similarly, PINQ can be simply extended to function video analytics. However, there are two major limitations of PINQ-based video analytics, which greatly constrain the usability in practices. 
\begin{itemize}
    \item \textbf{Sensitivity}. In PINQ-based video analytics, global sensitivity \cite{Dwork06} can be defined for some coarse-grained queries (with small sensitivity) such as ``the count of vehicles in the video'' (sensitivity as 1). However, for queries with large sensitivity (e.g., query across different frames where a video may include thousands of frames), the output result would be overly obfuscated (see experimental results in Section \ref{sec:exp}). For instance, in the query ``the average time (number of frames) each object stays in the video'', since an object can stay in the video for the entire video (all the frames) or only 1 second (a few frames), global sensitivity would be too large and difficult to define. Meanwhile, it might be also difficult to achieve (smooth) local sensitivity \cite{sensitivity07} for many different queries in the analysis due to computational overheads. 
    
    \item \textbf{Flexibility}. PINQ is inflexible to be adapted for utility-driven video analyses. For each requested analysis, a specific DP scheme would be required for improving the utility of the private analysis. The algorithm (e.g., budget allocation, composition of queries \cite{McSherry09}) has to be redesigned for any new analysis on the video. 
    
\end{itemize}

Instead, we propose a novel universal framework \emph{VideoDP} for universally optimizing the utility of different video analyses, which are detailed as follows.

\subsubsection{\emph{VideoDP} for Video Analytics}

Figure \ref{fig:overview} shows that \emph{VideoDP} consists of three major phases (after detecting all the VEs): 
\begin{enumerate}
\item 

Phase I: video (including detected VEs) can be represented as pixels, which can be grouped by their RGBs (notice that, different from generating RGB histograms, each pixel still keeps its original features such as coordinates and frame ID). Then, this phase samples pixels (with its original features) for each RGB where privacy budgets are allocated for different RGBs (\emph{sequential composition} \cite{McSherry09}) to optimize the output utility. Phase I in \emph{VideoDP} satisfies $\epsilon$-differential privacy. Details are given in Section \ref{sec:algm}.

\item Phase II: after sampling all the pixels, the output video has numerous unsampled pixels (due to privacy constraints). This phase estimates the RGBs for unsampled pixels via interpolation. We show that Phase II does not leak any additional information (still ensuring \emph{indistinguishability}). Thus, Phase II can boost the video utility without additional privacy loss. Details are given in Section \ref{sec:post}.

\item Phase III: for any query w.r.t. pixel, feature, visual elements, etc.\cite{imghistogram,Feature, Behavior},  \emph{VideoDP} applies the requested query to the \emph{random} utility-driven private video and directly returns the result (also probabilistic) to untrusted analysts. Then, \emph{VideoDP} can universally function any video analysis (which can be decomposed into queries, similar to PINQ \cite{McSherry09}) while ensuring differential privacy (as analyzed in Section \ref{sec:post}). 

\end{enumerate}

\section{Phase I: Pixel Sampling}
\label{sec:algm}
In this section, we illustrate the first phase of pixel sampling which satisfies $\epsilon$-differential privacy.

\subsection{Pixel Sampling Mechanism}

Recall that Section \ref{sec:frame} has briefly discussed the pixel sampling. For each RGB $\theta_i$ in $V$, a number of $x_i$ pixels (out of $c_i$ in the input $V$) will be randomly selected to output with their original coordinates and frame ID (\emph{while bounding the probabilities for differential privacy}). \emph{VideoDP} allocates privacy budgets for different RGBs to ensure differential privacy for the entire pixel sampling (e.g., $e^{-\epsilon}\leq \frac{Pr[\mathcal{A}(V)=O]}{Pr[\mathcal{A}(V')=O]}\leq e^{\epsilon}$).

After pixel sampling, all the RGBs $\theta_i, i\in[1,m]$ in the same visual element may result in privacy leakage in the output video. Then, the differential privacy of overall sampling follows \emph{sequential composition} \cite{McSherry09} for all the RGBs. Since every video may involve millions of distinct RGBs, given a privacy budget $\epsilon$ for pixel sampling, it is nearly impossible to allocate an equal budget to every unique RGB (each share would be negligible). To address such challenge, we categorize all the RGBs $i\in[1,m], \theta_i$ for pixel sampling in different cases (some of which indeed do not consume any privacy budget) and explore the optimal budget allocation as well as the differential privacy guarantee in Section \ref{sec:pri_con}. 

\subsection{Privacy Budget Allocation}
\label{sec:pri_con}
As the privacy budget $\epsilon$ is specified for pixel sampling, our goal is to optimize the allocated budgets for RGBs towards their count distributions in the original video. Given $V$ and $V'$ where $V=V'\cup \Upsilon$ (w.l.o.g.) and $\Upsilon$ can be any VE, we have three types of RGBs:

\begin{itemize}
    \item Case (1): RGB $\theta_i\in \Upsilon\setminus V'$ (the RGB is included in the extra visual element $\Upsilon$ but not $V'$).
    \item Case (2): RGB $\theta_i\in V'\setminus \Upsilon$ (the RGB is included in $V'$ but not the extra visual element $\Upsilon$).
    \item Case (3): RGB $\theta_i\in V'\cap \Upsilon$ (the RGB is included in both $V'$ and the extra visual element $\Upsilon$).
\end{itemize}

Then, we investigate the budget and the privacy guarantee for these three cases as below.

\subsubsection{Case (1): RGB $\theta_i\in \Upsilon\setminus V'$} 

Pixels in this case is the reason why we need the relax in definition, which we will discuss this in the Section \ref{sec:disc}.
Given $x_i$ as the output count of $\theta_i$ and $c_i$ is the input count in $V$, we let $x_i=0$ (does not output pixels with such RGB $\theta_i$) since $\theta_i$ cannot be found in $V'$, if generating any pixel with RGB $\theta_i$ into the output video $O$ (in Phase I).
 
Extending it to an randomization algorithm $\mathcal{A}$ applied to $V$ (with $n$ VEs $\Upsilon_1,\dots, \Upsilon_n$), w.l.o.g., considering $V$ as the video with an arbitrary extra VE $\Upsilon\in\{\Upsilon_1,\dots, \Upsilon_n\}$ (compared to $V'$), we have:

\begin{itemize}
\item $\forall j\in[1,n], \Upsilon_j$, if RGB $\theta_i\in \Upsilon_j\setminus (V-\Upsilon_j)$, then $x_i=0$ (do not sample pixels with such RGB). 
 
\end{itemize}

\subsubsection{Case (2): RGB $\theta_i\in V'\setminus \Upsilon$}
Since all the pixels with such RGB $\theta_i$ in $V$ and $V'$ are equivalent (identical coordinates and frame), we can let $x_i=c_i$ in \emph{VideoDP} (retaining all the pixels with such RGB $\theta_i$) without violating privacy. Then, for any $x_i>0$ (which can be maximized to $c_i$), sampling pixels for this RGB $\theta_i$ does not consume any privacy budget. Similarly, extending it to the randomization algorithm $\mathcal{A}$ (applied to $V$), w.l.o.g., considering $V$ as the video with an arbitrary extra VE $\Upsilon\in\{\Upsilon_1,\dots, \Upsilon_n\}$ (compared to $V'$), since \emph{VideoDP} should protect any arbitrary VE, we have:

\begin{itemize}
\item $\forall j\in[1,n], \Upsilon_j$, if any RGB $\theta_i\in V'\setminus \Upsilon_j$, then $x_i=c_i$ (retaining all the pixels with such RGB in the utility-driven private video). This does not consume any privacy budget since such RGBs do not exist in any of the VEs.

\end{itemize}

\subsubsection{Case (3): RGB $\theta_i\in V'\cap \Upsilon$}

The pixel sampling for each RGB in this case should satisfy $\epsilon$-differential privacy, and the overall sampling makes $e^{-\epsilon}\leq \frac{Pr[\mathcal{A}(V)=O]}{Pr[\mathcal{A}(V')=O]}\leq e^{\epsilon}$ hold. Thus, we should 
allocate privacy budgets for different RGBs in this case.  

However, due to the \emph{sequential composition} \cite{McSherry09}, we cannot allocate a budget for every RGB in this category (otherwise, given any $\epsilon$, for a large number of distinct RGBs, each share of the budget would be too extremely small). In other words, all the RGBs in this category may have to be suppressed (not sampled in the output video). To improve the output utility, our \emph{VideoDP} has the following three procedures for budget allocation in pixel sampling (Phase I): 

\begin{enumerate}
    \item Determine the RGBs selection rule (\emph{selecting the most representative RGBs in each VE for generating the output video}).
    \item Derive an optimal number of distinct RGBs within each VE (\emph{maximizing the utility of the VEs in the output video}).
    \item Allocate appropriate budgets for selected RGBs (\emph{per their RGB count distribution in the original video}).
\end{enumerate}

\noindent\textbf{1) RGBs Selection Rule.} Denoting the number of distinct RGBs in $\Upsilon_j, j\in[1,n]$ (which receive privacy budgets to output after Phase I) as $k_j$, the remaining RGBs in $\Upsilon_j$ will be suppressed (not sampled) during pixel sampling. Thus, this procedure ensures that the selected $k_j$ RGBs in $\Upsilon_j$ are most representative to reconstruct the object (without compromising privacy). An intuitive rule is to select the top frequent $k_j$ RGBs in $\Upsilon_j$. However, it might be biased to specific regions with intensive counts of similar RGBs in a VE. To address such limitation, we adopt the multi-scale analysis \cite{multi-scale08} in computer vision to partition each VE $\Upsilon_j$ into $k_j$ cells and select the top frequent RGB in each cell to allocate privacy budgets (as the ``representative RGBs''). Then, the sampled RGBs can be effective to reconstruct the VE in the utility-driven private video.

\vspace{0.05in}

\noindent\textbf{2) Optimal $k_j$ in Each VE.} This procedure is designed to maximize the utility of the VEs (i.e., object/human) in the utility-driven private video (after bilinear interpolation \cite{interpolation} in Phase II). If the number of distinct RGBs in $\Upsilon_j$ that receive privacy budgets $k_j$ is large, more distinct RGBs can be sampled in the VE but the budget allocated for each RGB would be extremely small; if $k_j$ is small, the budget allocated for each RGB would be large but less distinct RGBs can be sampled. We now seek for the optimal $k_j$ for $\Upsilon_j$ that can minimize the MSE between the interpolated VE (after Phase II) and the original VE. 

Specifically, since every pixel in $\Upsilon_j$ can be sampled (with the original RGB) or unsampled (with an estimated RGB), we minimize the expectation of MSE (referring to Equation \ref{eq:mse}) after the Phase II bilinear interpolation \cite{interpolation}. The expectation of each pixel's RGB is determined by the probabilities of ``sampled'' (denoted as $Pr(a,b,t)$) and ``unsampled and interpolated by its neighboring pixels'' (4 neighbors for a non-border pixel, 3 neighbors for a border-but-not-corner pixel, and 2 neighbors for a corner pixel, as shown in Figure \ref{fig:bilinear}. 

Denoting pixel $(a,b,t)$'s RGB in the output as $\hat{\theta}(a,b,t)$, for simplicity of notations, we denote the RGBs of its neighboring pixels as $\hat{\theta}_N, \hat{\theta}_S, \hat{\theta}_W$ and $\hat{\theta}_E$, for pixels $(a-1,b,t)$, $(a+1,b,t)$, $(a,b-1,t)$ and $(a,b+1,t)$, respectively. For a non-border pixel (4 neighbors), the expectation of its RGB \footnotemark[2] can be derived as:

\footnotetext[2]{Although the RGB values of all the pixels in $\Upsilon_j$ may be random (due to the differentially private sampling in Phase I), the expectations of RGBs for its neighboring pixels in $\Upsilon_j$ always satisfy a condition (ensured by bilinear interpolation \cite{interpolation}), e.g., Equation \ref{eq:mid}.}

\small
\begin{align}
&E[\hat{\theta}(a,b,t)]
=Pr(a,b,t)*\theta(a,b,t)+\sigma_0(a,b,t)*0\nonumber\\
+&\frac{\sigma_1(a,b,t)[1-Pr(a,b,t)][E(\hat{\theta}_N)+E(\hat{\theta}_S)+E(\hat{\theta}_W)+E(\hat{\theta}_E)]}{4}\nonumber\\
+&\frac{\sigma_2(a,b,t)[1-Pr(a,b,t)][3E(\hat{\theta}_N)+3E(\hat{\theta}_S)+3E(\hat{\theta}_W)+3E(\hat{\theta}_E)]]}{6*2}\nonumber\\
+&\frac{\sigma_3(a,b,t)[1-Pr(a,b,t)][3E(\hat{\theta}_N)+3E(\hat{\theta}_S)+3E(\hat{\theta}_W)+3E(\hat{\theta}_E)]}{4*3}\nonumber\\
+&\frac{\sigma_4(a,b,t)[1-Pr(a,b,t)][E(\hat{\theta}_N)+E(\hat{\theta}_S)+E(\hat{\theta}_W)+E(\hat{\theta}_E)]}{4}
\label{eq:mid}
\end{align}
\normalsize

where $\theta(a,b,t)$ is the original RGB (a constant) and probability of ``sampled'' $Pr(a,b,t)$ is determined by $k_j$ (given $V$ and $k_j$, it is deterministic if the RGB selection rule is decided previously). Probabilities $\sigma_0(a,b,t),\sigma_1(a,b,t), \sigma_2(a,b,t),\sigma_3(a,b,t)$ and $\sigma_4(a,b,t)$ are probabilities that pixel $(a,b,t)$ has 0 neighbor, 1 neighbor, 2 neighbors, 3 neighbors and 4 neighbors after sampling (which are also constants if $V$, $k_j$ and sampling mechanism are determined; note that $\sigma_0(a,b,t)+\dots+\sigma_4(a,b,t)=1$). In the equation, $E[\hat{\theta}_N]$, $E[\hat{\theta}_S]$, $E[\hat{\theta}_W]$ and $E[\hat{\theta}_E]$ are the RGB expectation of its four neighbors in the same $t$th frame (\emph{Equation \ref{eq:mid} presents the relation among the RGB expectations of the five pixels}, which are detailed in Appendix \ref{sec:border}). Similarly, we can obtain two other equations for pixels with special coordinates (border-but-not-corner or corner pixels of the frame, please see Equation \ref{eq:border} and \ref{eq:corner} in Appendix \ref{sec:border}). 

Thus, for each pixel in VE $\Upsilon_j$ (in all the frames), there exists exactly one equation out of three cases in Equation \ref{eq:mid}, \ref{eq:border} and \ref{eq:corner} (latter two are in Appendix \ref{sec:border}). As $k_j$ is determined, $\theta(a,b,t)$ and $Pr(a,b,t)$ are constants, then we can solve all the equations to obtain $\forall (a,b,t)\in \Upsilon_j, E[\hat{\theta}(a,b,t)]$. Thus, each $k_j$ value corresponds to the solved $\forall (a,b,t)\in \Upsilon_j, E[\hat{\theta}(a,b,t)]$, and then we can efficiently derive the optimal $k_j$ for $\Upsilon_j$ as:

\begin{equation}
    \argmin_{k_j}\frac{1}{|\Upsilon_j|}\sum_{\forall (a,b,t)\in \Upsilon_j}\big(E[\theta(a,b,t)]-E[\hat{\theta}(a,b,t)]\big)^2
    \label{eq:mse}
\end{equation}

where $|\Upsilon_j|$ denotes the total number of pixels in $\Upsilon_j$. Solving the above problem requires complexity $O(n^3\log(n))$, which is much faster than executing pixel sampling for all the possible $k_j$ and then comparing all the MSE results to get the optimal $k_j$ (\emph{since iteratively sampling all the pixels is expensive}). For details of the solver, please refer to Appendix \ref{sec:solver}. Notice that,

\begin{itemize}
\item \textbf{Range for $k_j$.} The optimal $k_j$ is derived from a specified range of $k_j$. It is unnecessary to traverse $k_j$ to a extremely large number (otherwise, the allocated budget for each RGB would be extremely small). The larger $k_j$, more diverse RGBs can be allocated with a privacy budget; the smaller $k_j$, each RGB will be allocated with a larger privacy budget. Thus, the lower/upper bounds for $k_j$ can be selected according the requested diversity of RGBs in the visual elements in practice ($k_j\leq 20$ can give good utility in our experiments). 

\item \textbf{Approximation.} As discussed before, since $\Upsilon_j$ in different frames may have different sizes and different sets of RGBs (though the difference can be minor), the most accurate $k_j$ can be obtained by solving the equations for all the pixels of $\Upsilon_j$ in all the frames (\emph{with complexity $O(n^3\log(n))$}, as proven in Appendix \ref{sec:opt}). If more efficient solvers are desirable, we can randomly select a frame (including $\Upsilon_j$) to solve the equations to obtain an approximated $k_j$ for $\Upsilon_j$ by assuming the VE does not change much in the video. Another alternative solution is to solve the optimal $k_j$ for each frame and average them (which is more efficient but less accurate).
\end{itemize}

Therefore, we can repeat the above procedure for all the VEs such that the optimal $k_j, j\in[1,n]$ can be obtained to minimize the MSE of the VEs in the utility-driven private video.

\vspace{0.05in}

\noindent\textbf{3) Budget Allocation.} As the optimal $k_j$ for each visual element $\Upsilon_j, j\in[1,n]$ is derived, we denote the set of RGBs in $\Upsilon_j, j\in[1,n]$ to allocate budgets as $\Psi_j$ with the cardinality $|\Psi_j|=k_j$. Then, we have the total number of RGBs to sample in $V$ (Case (3)) as the cardinality $|\Psi|$ of the union $\Psi=\bigcup_{j=1}^n\Psi_j$. We then present how to allocate privacy budget $\epsilon$ in Phase I for $|\Psi|$ different RGBs. The criterion for allocating budget is to allocate the privacy budgets based on the count distributions of RGBs in different VEs while fully utilizing the privacy budget $\epsilon$. For each VE $\Upsilon_j$, all the RGBs in $\Psi_j$ can fully enjoy the budget $\epsilon$ (since $\Psi_j$ includes all the RGBs that could generate visual element $\Upsilon_j$ in all the frames, and other RGBs would not be sampled into the visual element $\Upsilon_j$).\footnotemark[3]

\footnotetext[3]{Any two VEs (e.g., humans or objects) do not share pixels in the video since the front VE blocks a part of the back VE if they overlap in any frame.}

Then, we denote the $i$th RGB in $\Psi_j$ as $\widetilde{\theta}_{ij}$ where $i\in[1,k_j]$, and the count of $\widetilde{\theta}_{ij}$ in $\Upsilon_j$ as $d_j(\widetilde{\theta}_{ij})$ and the overall pixel count in $\Upsilon_j$ (in all frames) as $d_j$. Apparently, we can allocate $\frac{d_j(\widetilde{\theta}_{ij})\epsilon}{d_j}$ to RGB $\widetilde{\theta}_j(i), i\in[1,k_j]$ and apply this criterion to all the VEs. However, if any RGB $\widetilde{\theta}$ is included in multiple VEs (the intersections among the sets $\forall j\in[1,n], \Psi_j$), $\widetilde{\theta}$ will receive privacy budgets from different VEs (and should satisfy differential privacy for all of them). At this time, its budget should be allocated as the \emph{minimum} one out of all (otherwise, not all the VEs in pixel sampling can be protected with $\epsilon$-differential privacy since the budget for some VEs may exceed $\epsilon$). 

Nevertheless, if the minimum budget is adopted as above, some VEs cannot fully enjoy $\epsilon$ (the gap between $\widetilde{\theta}$'s original budget in a specific VE and its minimum budget among all the VEs would be wasted). To fully utilize the privacy budgets, we propose a \emph{budget allocation algorithm} for all the $|\Psi|$ distinct RGBs by \emph{prioritizing} them in the RGB set $\Psi=\bigcup_{j=1}^n\Psi_j$.

\begin{figure}[!h]
	\centering
		\includegraphics[angle=0, width=1\linewidth]{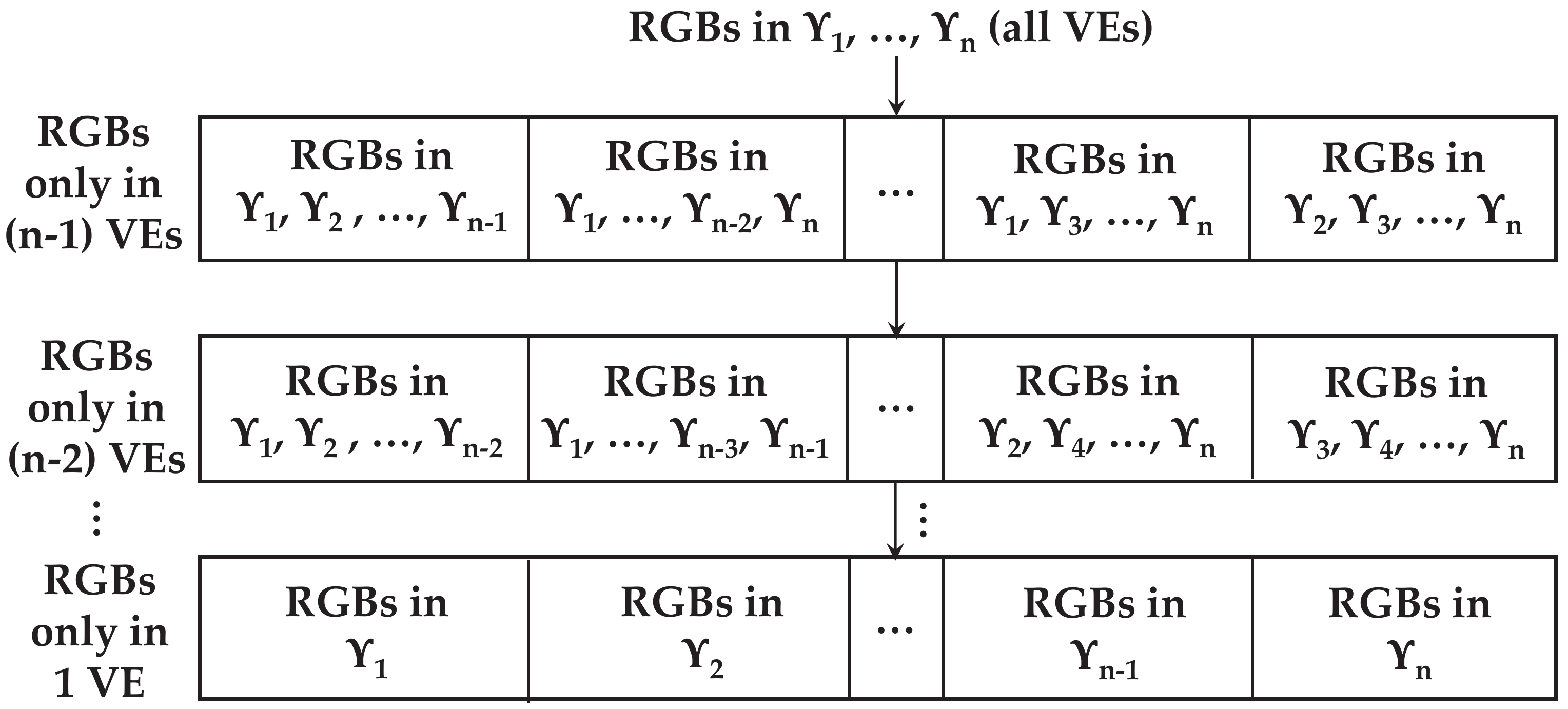}
	\caption[Optional caption for list of figures]
	{Prioritizing RGBs (for allocating budgets)}
	\label{fig:group}
\end{figure}

Specifically, we prioritize $|\Psi|$ different RGBs into $n$ disjoint partitions: as shown in Figure \ref{fig:group} (from top to down), RGBs in the first partition are included in all the VEs, RGBs in the second partition are included in $(n-1)$ VEs, \dots, RGBs in the $n$th partition are only included in a single VE. Then, our algorithm iteratively allocates budgets for RGBs in $n$ partitions (\emph{allocating budgets for all the RGBs in a partition in each iteration}). 

Since all the RGBs within each VE follow \emph{sequential composition} \cite{McSherry09} to split $\epsilon$, after allocating the budgets for all the RGBs in the $\ell$th partition, the allocation in the $(\ell+1)$th partition will be based on the remaining budget out of $\epsilon$ for every VE. In the $\ell$th iteration (for the $\ell$th partition), the budget for each RGB $\widetilde{\theta}$ is allocated based on its count distribution out of the remaining RGBs in each of the $(n-\ell+1)$ VEs (which include $\widetilde{\theta}$). Then, the \emph{minimum} budget derived from all the VEs is allocated to $\widetilde{\theta}$. 

\begin{figure}[!tbh]
	\centering
		\includegraphics[angle=0, width=1\linewidth]{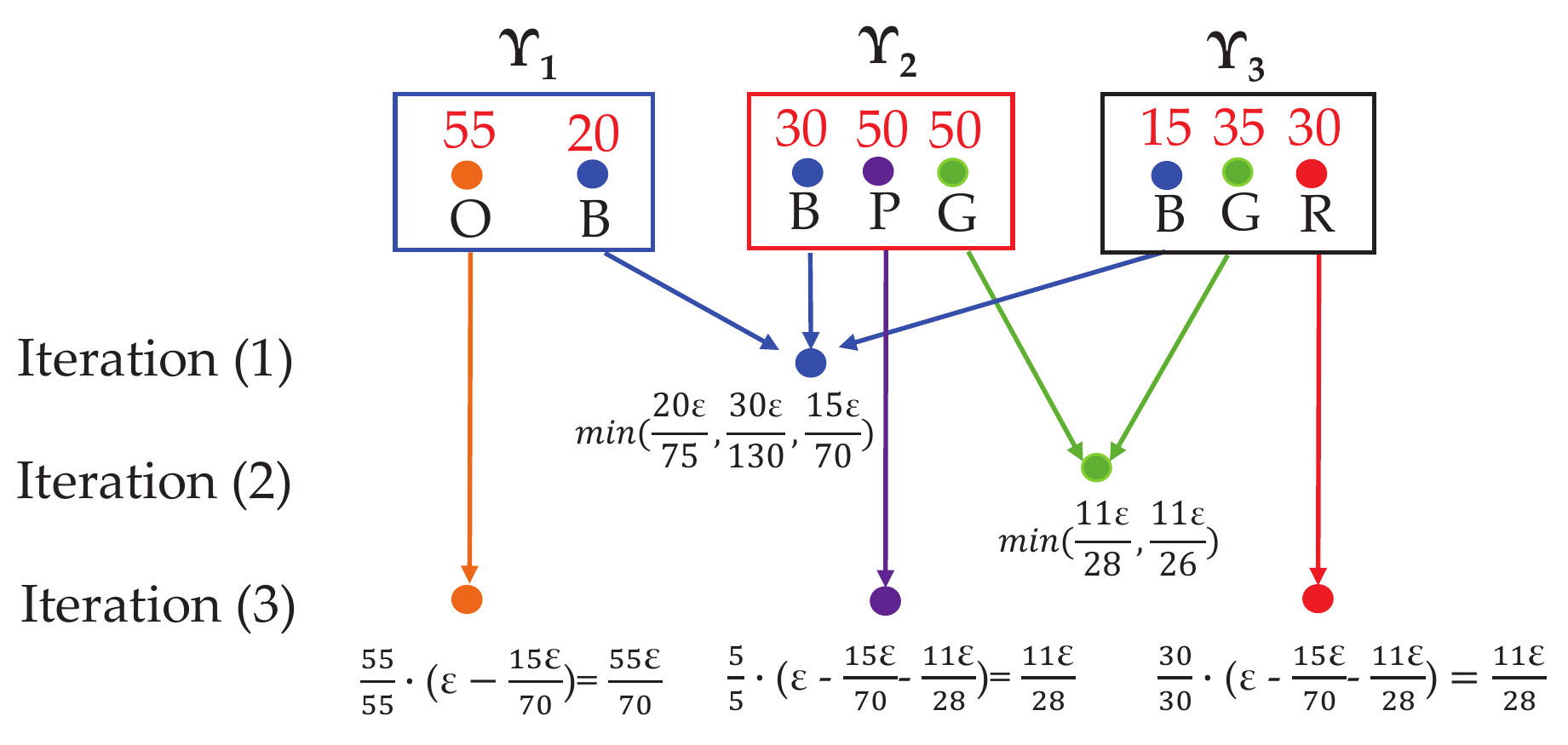}
		
	\caption[Optional caption for list of figures]
	{Example of Budget Allocation}
	\label{fig:example}
\end{figure}

\begin{example}
In Figure \ref{fig:example}, there are three VEs $\Upsilon_1,\Upsilon_2,\Upsilon_3$ in the video. Blue exists in all the VEs $\Upsilon_1,\Upsilon_2,\Upsilon_3$ with counts $20, 30, 15$. Green exists in $\Upsilon_2$ and $\Upsilon_3$ with counts $50$ and $35$. All the remaining RGBs only exist in only one VE (and the non-VE part of the video): Orange in $\Upsilon_1$ with count $55$, Purple in $\Upsilon_2$ with count $5$, and Red in $\Upsilon_3$ with count $30$. Thus, five different RGBs are prioritized (three partitions): \{B\}, \{G\}, and \{O, P, R\}.

In the 1st iteration (partition), Blue is first allocated with a privacy budget as the $\min\{\frac{20\epsilon}{75}, \frac{30\epsilon}{130}, \frac{15\epsilon}{70}\}$ (the minimum budget from three different VEs). The remaining budget for all the VEs is $\frac{55\epsilon}{70}$. In the 2nd iteration, Green is allocated with a privacy budget $\frac{55\epsilon}{70}\cdot \min\{\frac{50}{100}, \frac{35}{65}\}=\min\{\frac{11\epsilon}{28}, \frac{11\epsilon}{26}\}$. In the 3rd iteration, Orange is allocated with budget $\frac{55}{55}\cdot (\epsilon- \frac{15\epsilon}{70})=\frac{55\epsilon}{70}$, Purple is allocated with budget $\frac{5}{5}\cdot(\epsilon-\frac{15\epsilon}{70}-\frac{11\epsilon}{28})=\frac{11\epsilon}{28}$, and Red is allocated with budget $\frac{30}{30}\cdot(\epsilon-\frac{15\epsilon}{70}-\frac{11\epsilon}{28})=\frac{11\epsilon}{28}$.
  
\end{example}

Since almost all the VEs have RGBs in the last partition (every VE in real videos include numerous RGBs that are not included in other VEs), the budget can be fully allocated for all the RGBs (\emph{the budget sum of all the RGBs} in any VE equals $\epsilon$). Algorithm \ref{algm:budget} in Appendix \ref{sec:baa} presents the details of budget allocation.

\subsection{Pixel Sampling Algorithm}

To illustrate the algorithm for Phase I, we again discuss the pixel sampling for three different cases of RGBs. 

Recall that in Case (1), for all the RGBs $\theta_i\in \Upsilon\setminus V'$, all the pixels with such RGBs will not be sampled (ensuring that $\delta=0$). In Case (2), for all the RGBs $\theta_i\in V'\setminus \Upsilon$, all the pixels with such RGBs will be sampled (with the original coordinates and frame). Sampling pixels for all the RGBs in Case (2) satisfy $0$-DP.

In Case (3), for all the RGBs $\theta_i\in V'\cap \Upsilon$, as discussed in Section \ref{sec:pri_con}, we sample pixels for $|\Psi|$ distinct RGBs where $|\Psi|\leq \sum_{j=1}^nk_j$ (since different VEs may have common RGBs). We denote the set $\Psi=\bigcup_{j=1}^n\Psi_j=\{\widetilde{\theta}_1, \dots, \widetilde{\theta}_{|\Psi|}\}$ (the set of RGBs which request privacy budgets), and its set of budgets $\{\epsilon(\widetilde{\theta}_1), \dots, \epsilon(\widetilde{\theta}_{|\Psi|})\}$. It is straightforward to show the \emph{sequential composition} \cite{McSherry09} of allocated privacy budgets (by Algorithm \ref{algm:budget}) for all the RGBs: 

\begin{equation}
\small
    \sum_{\forall \widetilde{\theta}_i\in \Psi_j}\epsilon(\widetilde{\theta}_i)=\epsilon
\end{equation}

where $\widetilde{\theta}_i$ is denoted as the $i$th RGB in $\Psi$. Then, for any $V$ and $V'$ differing in an arbitrary VE $\Upsilon_j, j\in[1,n]$, 

\begin{equation}
\forall \widetilde{\theta}_i\in \Psi_j, e^{-\epsilon(\widetilde{\theta}_i)}\leq\frac{Pr[\mathcal{A}(V(\widetilde{\theta}_i))=O(\widetilde{\theta}_i)]}{Pr[\mathcal{A}(V'(\widetilde{\theta}_i))=O(\widetilde{\theta}_i)]}\leq e^{\epsilon(\widetilde{\theta}_i)}
\end{equation}

where $V(\widetilde{\theta}_i)$ and $V'(\widetilde{\theta}_i)$ are the pixels with RGB $\widetilde{\theta}_i$ in $V$ and $V'$. Deriving the probability for randomly picking $\widetilde{x}_i$ out of $\widetilde{c}_i$ pixels with RGB $\widetilde{\theta}_i$ (pixel sampling using input $V$ and $V'$, differing in $\Upsilon_j$), we have: 
\begin{align}
\forall i\in[1,|\Psi|],~&Pr[\mathcal{A}(V(\widetilde{\theta}_i))=O(\widetilde{\theta}_i)]=1/\binom{\widetilde{c}_i}{\widetilde{x}_i}\nonumber\\
&Pr[\mathcal{A}(V'(\widetilde{\theta}_i))=O(\widetilde{\theta}_i)]=1/\binom{\widetilde{c}_i-\widetilde{c}_i^j}{\widetilde{x}_i}\nonumber\\
\implies &e^{-\epsilon(\widetilde{\theta}_i)}\leq\binom{\widetilde{c}_i}{\widetilde{x}_i}\big{/}\binom{\widetilde{c}_i-\widetilde{c}_i^j}{\widetilde{x}_i}\leq e^{\epsilon(\widetilde{\theta}_i)}
\end{align}

where $\widetilde{c}_i$ and $\widetilde{x}_i$ are the input and output counts of RGB $\widetilde{\theta}_i$ while $\widetilde{c}_i^j$ denotes the count of $\widetilde{\theta}_i$ in VE $\Upsilon_j$. 

Thus, we can derive a maximum output count for sampling pixels for each RGB $\widetilde{\theta}_i, i\in[1,|\Psi|]$ and the maximum $\widetilde{x}_i$ can be efficiently computed as below (the only variable): $\forall i\in[1,|\Psi|]$,
\begin{equation}
\max \{\widetilde{x}_i| \forall j\in [1,n], \binom{\widetilde{c}_i}{\widetilde{x}_i}\big{/}\binom{\widetilde{c}_i-\widetilde{c}_i^j}{\widetilde{x}_i}\leq e^{\epsilon(\widetilde{\theta}_i)}\}
\label{eq:max}
\end{equation}

The maximum output count of the $i$th RGB $\widetilde{x}_i, i\in[1,|\Psi|]$ can be efficiently computed from Equation \ref{eq:max} (e.g., via binary search) since the left-side of the inequality is monotonic on $\widetilde{x}_i$. To sum up, Algorihtm \ref{algm:pixel} presents the details of Phase I.

\begin{algorithm}[!h]
\small
\SetKwInOut{Input}{Input}\SetKwInOut{Output}{Output}

\Input{input video $V$, privacy budget $\epsilon$}
 
\Output{sampled video $O$ (pixels)}

detect all the visual elements $(\Upsilon_1,\dots, \Upsilon_n)$ in $V$

\tcp{Case (1)}

\ForEach{$\Upsilon_j, j\in[1,n]$}{
\ForEach{RGB $\theta_i\in \Upsilon_j$ but $\notin (V\setminus \Upsilon_j)$}{

suppress all $c_i$ pixels with RGB $\theta_i$ in $V$  
}
}

\tcp{Case (2)}

\ForEach{RGB $\theta_i\in V\setminus\bigcup_{j=1}^n\Upsilon_j$}{
output all $c_i$ pixels with RGB $\theta_i$ in $V$ (original coordinates and frame)}

\tcp{Case (3)}

\ForEach{$\Upsilon_j, j\in[1,n]$}{
compute the optimal number of distinct RGBs to sample in $\Upsilon_j$ (minimum expectation of MSE): $k_j$}

execute Algorithm \ref{algm:budget} (in Appendix \ref{algm:budget}) to allocate budgets for all the RGBs in $\Psi=\{\widetilde{\theta}_1, \dots, \widetilde{\theta}_{|\Psi|}\}$

\ForEach{$\widetilde{\theta}_i, i\in[1,|\Psi|]$}{

compute the maximum $\widetilde{x}_i$:
$\max \{\widetilde{x}_i| \forall j\in [1,n], \binom{\widetilde{c}_i}{\widetilde{x}_i}\big{/}\binom{\widetilde{c}_i-\widetilde{c}_i^j}{\widetilde{x}_i}\leq e^{\epsilon(\widetilde{\theta}_i)}\}$

randomly pick $\widetilde{x}_i$ pixels with RGB $\widetilde{\theta}_i$ in $V$ to output (original coordinates and frame)

}
	\caption{Pixel Sampling ($\epsilon$-DP)}
	\label{algm:pixel}
\end{algorithm}    

\begin{theorem}
The pixels sampling in \emph{VideoDP} (Phase I) satisfies $\epsilon$-differential privacy.
\label{theorem:dp}
\end{theorem}

\begin{proof}
We can prove the differential privacy guarantee for three cases of pixel sampling in the algorithm. 

In Case (1), since all the pixels with such RGBs are suppressed, $\delta=0$ always holds with Line 2-4 in Algorithm \ref{algm:pixel}. In Case (2), since $\forall \theta_i, \frac{Pr[\mathcal{A}(V(\theta_i))=O(\theta_i)]}{Pr[\mathcal{A}(V'(\theta_i))=O(\theta_i)]}$ always equals $1$, Line 5-6 in Algorithm \ref{algm:pixel} does not result in privacy loss. In Line 7-12 of the algorithm (Case (3)), we have $\forall i\in[1,|\Psi|]$, $e^{-\epsilon(\widetilde{\theta}_i)}\leq\frac{Pr[\mathcal{A}(V(\widetilde{\theta}_i))=O(\widetilde{\theta}_i)]}{Pr[\mathcal{A}(V'(\widetilde{\theta}_i))=O(\widetilde{\theta}_i)]}\leq e^{\epsilon(\widetilde{\theta}_i)}$ holds. Per the \emph{sequential composition} of differential privacy \cite{McSherry09}, for all $V$ and $V'$ differing in any VE $\Upsilon_j, j\in[1,n]$, we have:

\small

\begin{align}
    &\prod_{\forall \widetilde{\theta}_i\in \Psi_j}\frac{Pr[\mathcal{A}(V(\widetilde{\theta}_i))=O(\widetilde{\theta}_i)]}{Pr[\mathcal{A}(V'(\widetilde{\theta}_i))=O(\widetilde{\theta}_i)]}\leq exp[\displaystyle\sum_{\forall \widetilde{\theta}_i\in \Psi_j}\epsilon(\widetilde{\theta}_i)]\nonumber
\end{align}

\begin{align}
        &\prod_{\forall \widetilde{\theta}_i\in \Psi_j}\frac{Pr[\mathcal{A}(V(\widetilde{\theta}_i))=O(\widetilde{\theta}_i)]}{Pr[\mathcal{A}(V'(\widetilde{\theta}_i))=O(\widetilde{\theta}_i)]}\geq exp[\displaystyle-\sum_{\forall \widetilde{\theta}_i\in \Psi_j}\epsilon(\widetilde{\theta}_i)]\nonumber\\
    &\implies e^{-\epsilon}\leq\frac{Pr[\mathcal{A}(V)=O]}{Pr[\mathcal{A}(V')=O]}\leq e^{\epsilon}
\end{align}

\normalsize

Thus, this completes the proof. 
\end{proof}

\noindent\textbf{Discussion.} In case of $V'=V \cup \Upsilon$, adding an arbitrary VE $\Upsilon$ to $V$ to generate $V'$. Similarly, for all $\widetilde{\theta}_i$, $\widetilde{x}_i$ can also be derived from $V'$ and $V$ to ensure differential privacy for pixel sampling.

\section{Phase II and Phase III}
\label{sec:post}
\pdfoutput=1
After sampling pixels in Phase I, the suppressed pixels in Case (1) and unsampled pixels in Case (3) in the output do not have any RGB (as shown in Figure \ref{fig:bilinear}). Then, Phase II generates the utility-driven private video (random) by estimating the RGBs for the missing pixels to boost utility with computer vision techniques, and Phase III responds to the queries (over the private video) for video analytics.

\subsection{Phase II: Generating Utility-driven Private Video (Random)}

\begin{figure}[!tbh]
	\centering
		{\includegraphics[angle=0, width=0.8\linewidth]{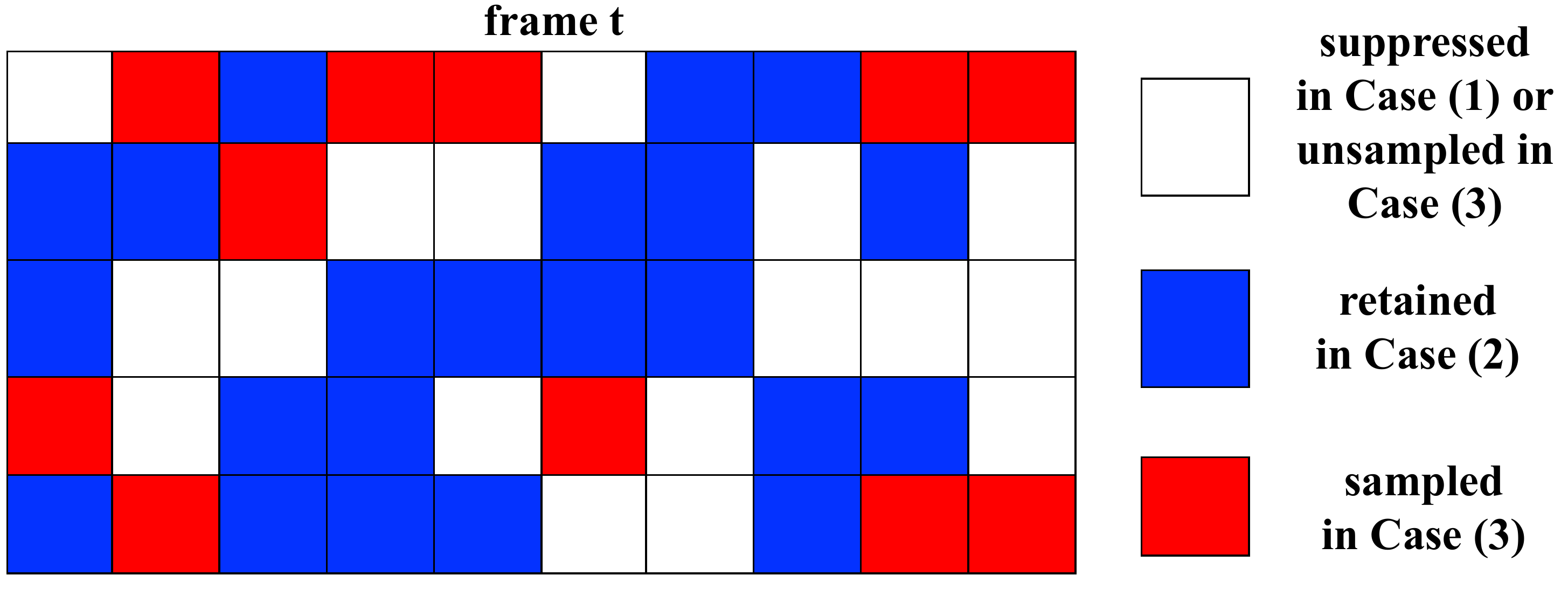}}
	\caption[Optional caption for list of figures]
	{Pixels after Sampling (Phase I)}
	\label{fig:bilinear}
\end{figure}

For all the coordinates with a RGB value after sampling, the RGBs of such pixels can be estimated using bilinear interpolation \cite{interpolation}. As discussed in Section \ref{sec:pri_con}, the allocated privacy budgets have been shown to optimize the utility of both sampling and bilinear interpolation, e.g., the optimal number of RGBs selected in each VE for sampling $k_j$ tends to minimize the expectation of MSE between the utility-driven private video (after interpolation) and the original video. Thus, Phase II can directly apply bilinear interpolation.

For simplification of notations, we consider both retained pixels and sampled pixels as ``sampled pixels'', and both suppressed pixels and unsampled pixels as ``unsampled pixels''. Specifically,

\begin{itemize}
\item In the output video of Phase I, pixels (not on the border) have at most 4 neighbors in each frame; the pixels on the border of each frame (not corner) have at most 3 neighbors; the pixels at the corner of each frame have at most 2 neighbors.

\item The algorithm interpolates pixels in visual elements and the remaining pixels (background), separately. For each interpolation, it traverses all the unsampled pixels in all the frames (e.g., a specific visual element). If any unsampled pixel has any sampled neighbor(s), the RGB for current unsampled pixel is estimated as the \emph{mean} of all its sampled neighbors.

\item If any unsampled pixel's all the neighbors are also unsampled, the algorithm skips such unsampled pixel in the current traversal. The algorithm iteratively traverses all the skipped unsampled pixels. The algorithm terminates until every unsampled pixel is assigned with an interpolated RGB. In our experiments, the interpolation terminates very quickly since the RGB of any pixel can be readily estimated as long as it has at least one neighbor which is sampled or previously interpolated. 
\item If any visual element does not have a sampled pixel in any frame, the interpolation of the pixels for the visual element in such frame will be executed with the remaining pixels (background)  $V\setminus \bigcup_{j=1}^n\Upsilon_j$.

\end{itemize}

Algorithm \ref{algm:inter} presents the details of pixel interpolation. Notice that, besides interpolating pixels in each frame, Phase II also  interpolates the VEs in specific frames (if none of their pixels are sampled in the frames but they are sampled in the neighboring frames). Such VEs will be inserted into the corresponding frames where the coordinates/RGBs are estimated by averaging their interpolated results in the neighboring frames.

\begin{algorithm}[!h]
\small 
\SetKwInOut{Input}{Input}\SetKwInOut{Output}{Output}

\Input{output video $O$ in Phase I}
 
\Output{utility-driven private video $\mathbb{O}$}

\tcp{interpolating pixels in VEs}

\ForEach{VE $\Upsilon_j, j\in[1,n]$}{

extract all the unsampled pixels in $\Upsilon_j$ with their frame, coordinates, and RGBs ($U\leftarrow \Upsilon_j\setminus O$)

\Repeat{$U=\emptyset$}{

\ForEach{unsampled pixel $p\in U$}
{

\If{$p$'s all the neighbors in $U$}{

\textbf{continue}

}\Else{

calculate the RGB mean of $p$'s sampled neighbors (in $O$) and assign it to $p$

$\mathbb{O}\leftarrow O\cup p$ and $U\leftarrow U\setminus p$

}

}

}

}

\tcp{interpolating background pixels}

repeat Line 1-10 for all the pixels in $V\setminus \bigcup_{j=1}^n\Upsilon_j$ 

	\caption{Pixel Interpolation}
	\label{algm:inter}
\end{algorithm}

\subsection{Video Analytics and Privacy Analysis}
\label{sec:pa}

Similar to the framework of PINQ for data analytics \cite{McSherry09}, \emph{VideoDP} can also function most of the analyses performed on videos. If breaking down any video analysis into queries, \emph{VideoDP} (Phase III) directly applies the queries to the \emph{utility-driven private video} (which is randomly generated in Phase I and II) and return the results to untrusted analysts. For any query created at the pixel, feature or visual element level \cite{imghistogram,Feature, Behavior}, \emph{VideoDP} (Phase III) could efficiently respond the results with differential privacy guarantee.

\begin{theorem}
\emph{VideoDP} satisfies $\epsilon$-differential privacy.
\label{theorem:qdp}
\end{theorem}
\begin{proof}

Recall that we have proven Phase I satisfies $\epsilon$-differential privacy in Theorem \ref{theorem:dp}. We now prove that Phase II and III do not result in additional privacy risks.

Since Phase I in \emph{VideoDP} satisfies $\epsilon$-differential privacy, for any pair of neighboring videos $V$ and $V'$, we have $e^{-\epsilon}\leq\frac{Pr[\mathcal{A}(V)=O]}{Pr[\mathcal{A}(V)=O]}\leq e^\epsilon$. Such differential privacy satisfies $\epsilon$-probabilistic differential privacy \cite{MachanavajjhalaKAGV08,corn09}, which also satisfies $\epsilon$-indistinguishability differential privacy \cite{Dwork06differentialprivacy,Dwork06} (bounding $Pr[\mathcal{A}(V)\in S]$ and $Pr[\mathcal{A}(V')\in S]$ where $S$ is any set of possible outputs), as proven in Proposition \ref{thm:pdp} in Appendix \ref{sec:indist} \cite{MachanavajjhalaKAGV08,corn09}.

Then, after applying \emph{VideoDP} to inputs $V$ and $V'$, the outputs of Phase I are $\epsilon$-indistinguishable. Since the pixel interpolation (Phase II) and video queries/analysis (Phase III) are deterministic procedures applied to the output of Phase I (which can be considered as \emph{post-processing} differentially private results), the output $\mathbb{O}$ of Phase II and the analysis/query results of Phase III derived from inputs $V$ and $V'$ are also $\epsilon$-indistinguishable (``\emph{Differential privacy is immune to post-processing}'' has been proven in \cite{Dwork14}). Therefore, \emph{VideoDP} also satisfies $\epsilon$-differential privacy.
\end{proof}

The procedures and privacy guarantee in \emph{VideoDP} can be interpreted as follows. Given any two videos $V$ and $V'$ that differ in one VE (e.g, a pedestrian), Phase I and II generate $\epsilon$-indistinguishable utility-driven private video (which is random) for $V$ and $V'$, respectively. Performing any query (e.g., the count of pedestrians) over the two indistinguishable utility-driven private video, the query/analysis results are also indistinguishable.

\section{Discussion}
\label{sec:disc}
\noindent \textbf{Relaxed Differential Privacy}: in general, Theorem \ref{theorem:qdp} guarantees that the query/analysis results satisfy $\epsilon$-DP. However, in some extreme cases, if real query result in video $V$ equals one and the result in $V'$ is zero (e.g., query is only related to the extra VE), we have to adopt $(\epsilon,\delta)$-differential privacy (Definition \ref{def:rdp}) in \emph{VideoDP}. In such cases, the pixel sampling in \emph{VideoDP} (Phase I) makes the output of $V$ to be $\epsilon$-indisintiguishable as the output of $V'$ where the extra VE is suppressed with high probability ($\geq 1-\delta$). Then, such query over the utility-driven private video derived by $V$ would return 0 with high probability (still ensuring indistinguishability and differential privacy). For simplicity of cases, we only focus on generic queries/analysis in this paper.

\vspace{0.05in}
\noindent\textbf{Non-Sensitive Visual Elements}: before generating videos by \emph{VideoDP}, sensitive visual elements are detected and specified for differentially private protection. If some objects and/or humans are unnecessary to protect (e.g., considered as non-sensitive), \emph{VideoDP} can leave them with the background scene and retain more utility for them in the utility-driven private video.

\vspace{0.05in}
    
\noindent \textbf{Arbitrary Background Knowledge}: differential privacy ensures protection/indistinguishability against arbitrary background knowledge. \emph{VideoDP} is proposed to only protect the visual elements (e.g., sensitive objects and humans). Thus, the background knowledge is assumed to be any information on the visual elements which are detected/specified to be protected.

\vspace{0.05in}

    \noindent\textbf{Defense against Correlations}: videos include numerous frames, if protecting specific visual elements in only one frame, the correlations in sequential frames may also leak information to adversaries \cite{SongSigmod17,CaoICDE17}. Our \emph{VideoDP} can address such vulnerabilities since all the visual elements in all the frames are protected using our defined privacy notion -- adding or removing \emph{any visual element in any number of frames} would not result in significant risks. From this perspective, the privacy notion is defined for the entire period of the video, rather than a specific time. Thus, possible privacy leakage resulted from correlations among multiple frames can be tackled.

\section{Experiments}
\label{sec:exp}

\begin{figure*}[!tbh]
	\centering
	\subfigure[KL vs $\epsilon$ (Video PED)]{
		\includegraphics[angle=0, width=0.249\linewidth]{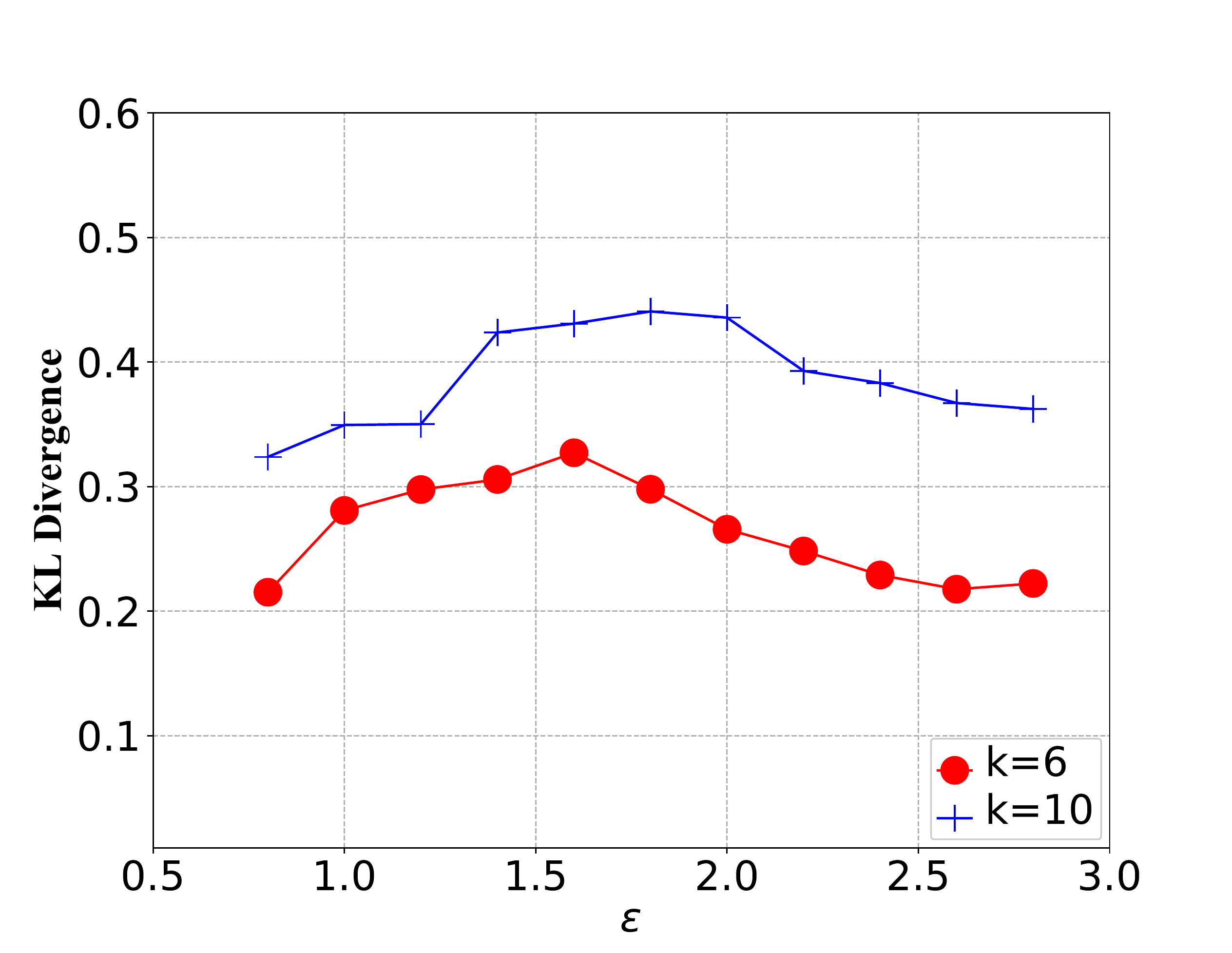}
		\label{fig:kle1} }
		\hspace{-0.2in}
		\subfigure[KL vs $\epsilon$ (Video VEH)]{
		\includegraphics[angle=0, width=0.249\linewidth]{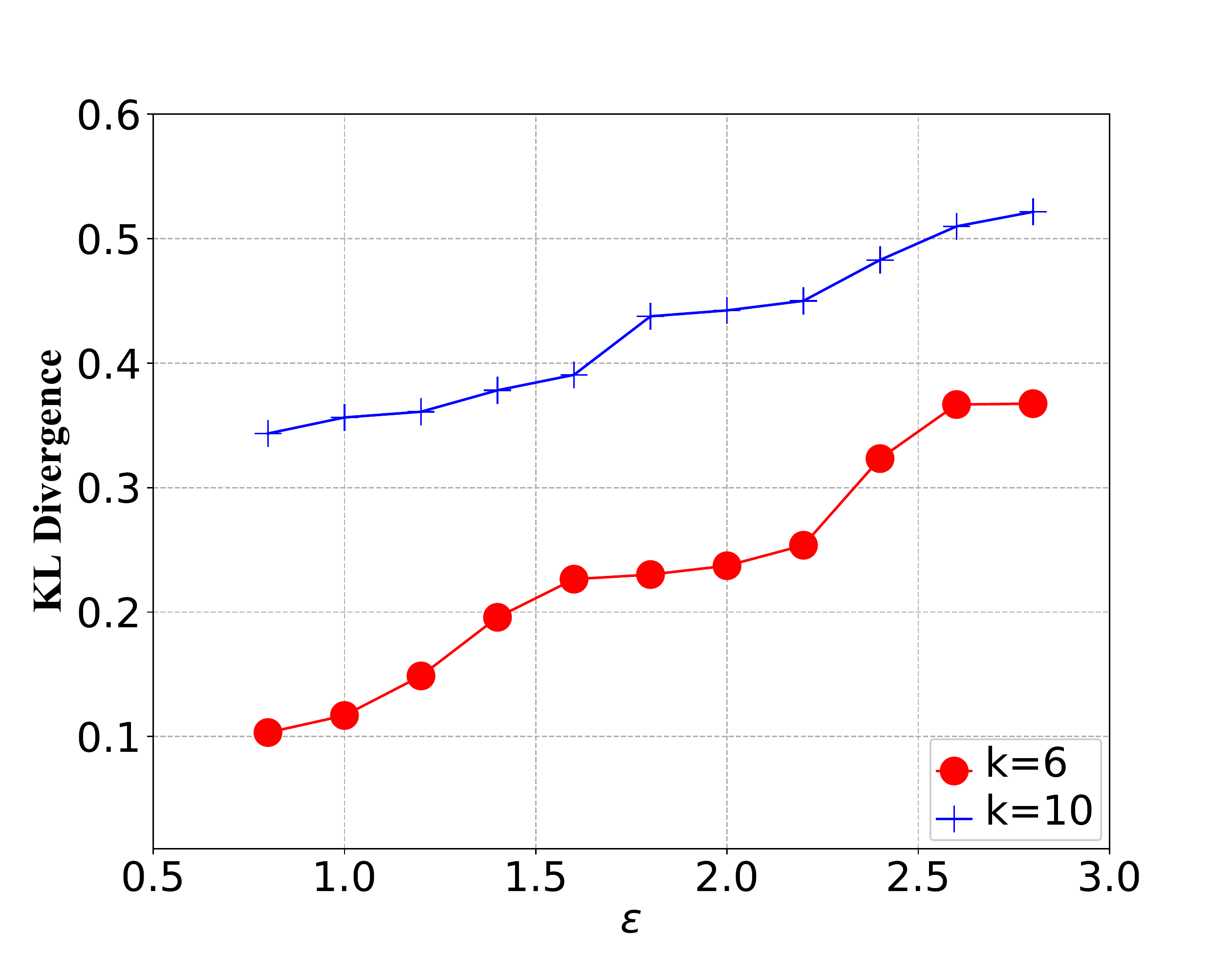}
		\label{fig:kle2} }
		\hspace{-0.2in}
	\subfigure[MSE vs $\epsilon$ (after Phase I)]{
		\includegraphics[angle=0, width=0.249\linewidth]{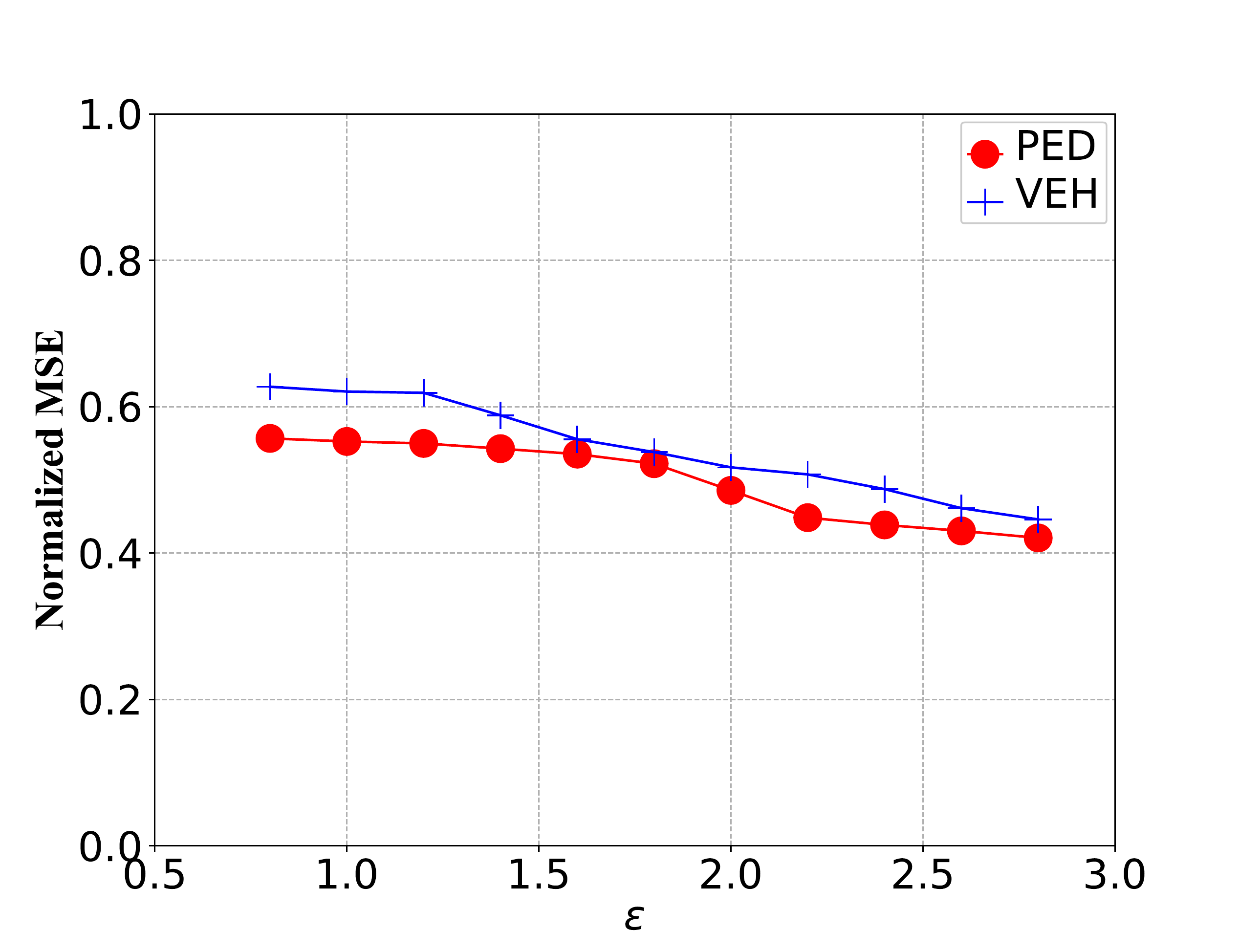}
		\label{fig:mseebefore} }
		\hspace{-0.2in}
	\subfigure[MSE vs $\epsilon$ (after Phase II)]{
		\includegraphics[angle=0, width=0.249\linewidth]{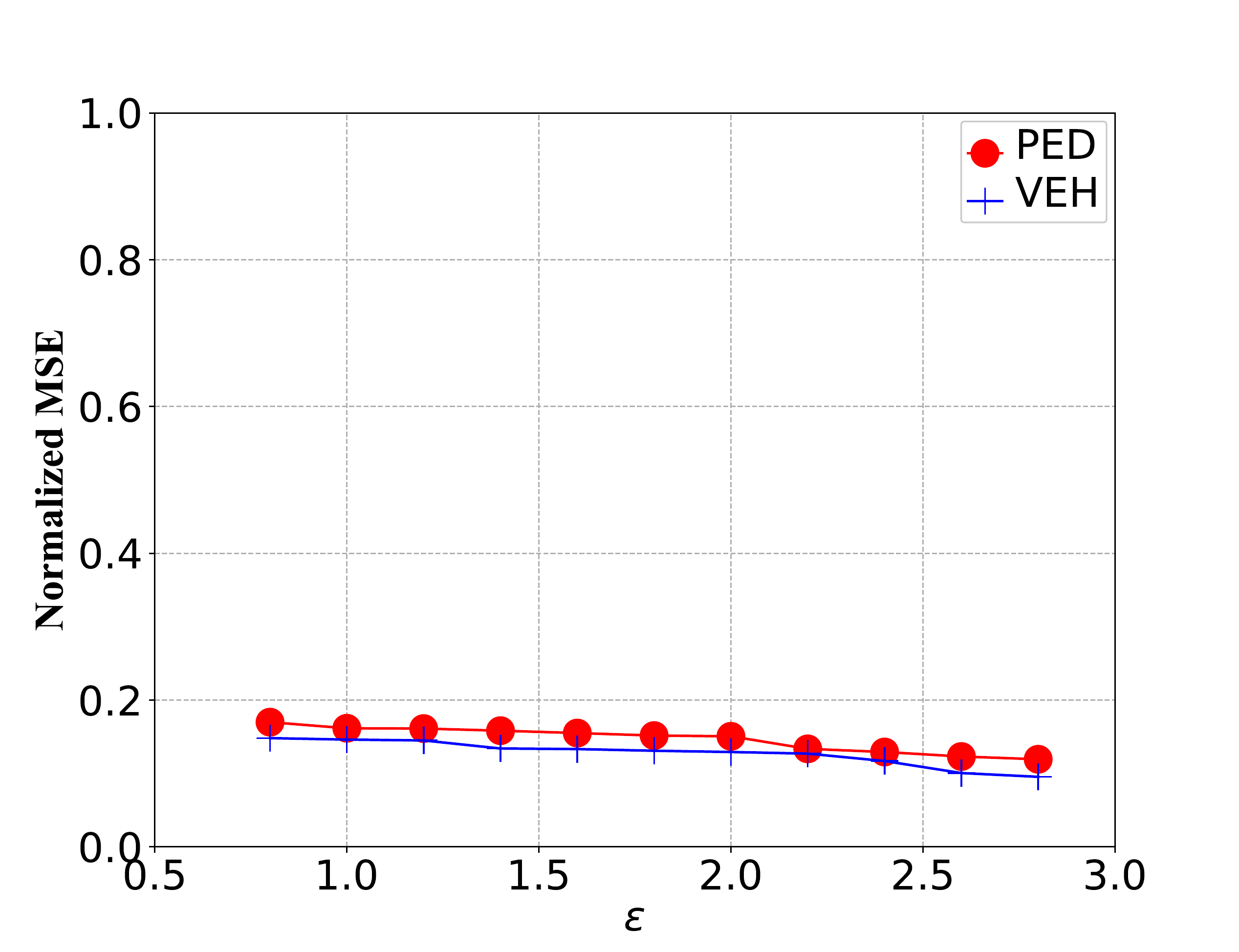}
		\label{fig:mseeafter}}
	
	\vspace{-0.05in}
	\caption[Optional caption for list of figures]
	{Pixel Level Utility Evaluation }
	\label{fig:KL}
\end{figure*}

In \emph{VideoDP}, we implement the detection of visual elements (VEs) in the entire video using the tracking algorithm \cite{tracking1,tracking2}, which first detects all the VEs in each frame, and then utilizes the tensorflow training database to tag humans/objects which are considered as sensitive VEs in different scenarios/videos. 
Each detected human/object can be tracked with the same ID if their overlap in multiple frames has exceeded a threshold (\emph{such algorithm has a high detection/tracking accuracy \cite{tracking1}}). We conduct our experiments on two real high-resolution videos in which different VEs are protected (characteristics are presented in Table \ref{table:data}).

\begin{table}[!h]
	\centering
	\caption{Characteristics of Experimental Videos}	
	\vspace{-0.1in}
		\begin{tabular}{|c|c|c|c|c|}	
			\hline
		Video & Resolution & Frame \# & RGBs  & VEs\\
		\hline\hline
			PED & $1920\times1080$ & 1,050  &546,430 &83  \\\hline
			VEH & $1280\times720$ &  540 &290,172 &115\\\hline

		\end{tabular}
		\vspace{-0.15in}
	\label{table:data}
\end{table}

\begin{enumerate}
    \item MOT16-04 (pedestrian street) \cite{MOT16}: 83 pedestrians are considered as sensitive visual elements in 1,050 frames (546,430 distinct RGBs). For simplicity of notations, we denote this video as ``PED''. 
    
    \item Vehicle Detection Video \cite{car}: 115 vehicles in the video are considered as sensitive visual elements in 540 frames (290,172 distinct RGBs). For simplicity of notations, we denote this video as ``VEH''.

\end{enumerate}

All the programs were implemented in Python 3.6.4 with OpenCV 3.4.0 library \cite{opencv} and tested on an HP PC with Intel Core i7-7700 CPU 3.60GHz and 32G RAM. 
\subsection{Evaluating Utility-driven Private Video}
\label{sec:utility}

We first evaluate the utility of the videos randomly generated by Phase I and II which directly reflects the universal utility for different video analyses. We consider the RGB color model \cite{Acharya:2005:IPP:1088917} by breaking down the videos into pixels with RGBs at different coordinates and frames, and then measure the differences between input $V$ and output $O$. Specifically, we evaluate two types of utility: (1) the difference between the count distributions of all the RGBs in $V$ and $O$, and (2) the difference between RGB values of all the pixels in $V$ and $O$.

First, considering the distributions of all the RGBs' counts $\forall x_i$ and $\forall c_i$ in the output and input, we can measure the utility loss using their KL-divergence for the following reason.

\begin{itemize}
    \item If the distribution of RGBs lie closes in the input and output videos, the performance of pixel interpolation (estimating RGBs for unsampled pixels based on the RGBs of sampled pixels) can be greatly improved \cite{interpolation}. For other measures, e.g., $L_1$ norm, the output counts of different RGBs might be biased towards certain RGBs with high counts such that the interpolated RGBs might be significantly deviated.
\end{itemize}

Thus, we adapt the KL-divergence based utility loss function in our \emph{VideoDP} as:

\begin{equation}
\small
D_{KL}=\sum_{i=1}^m
\big{[}\frac{c_i}{|V|}\cdot\log(\frac{c_{i}}{|V|}\cdot\frac{|O|}{x_i+1})\big{]}
\label{eq:wkl}
\end{equation}

where $m$ denotes the number of distinct RGBs in the input $V$, and $|V|$ and $|O|$ denote the total pixel counts in the input $V$ and output $O$. Moreover, $c_i$ denotes the pixel count with RGB $\theta_i$ in $V$ whereas $x_i$ represents such count in $O$. Since $x_i$ may equal 0, we use $(x_i+1)$ to replace $x_i$ (for avoiding zero-denominator) in which $(x_i+1)$ lies very close to $x_i$ in the context of videos. Then, we use the above function to evaluate the utility loss in \emph{VideoDP} (\emph{the privacy budget is also allocated by following the distribution of RGB counts}).

Second, after interpolating all the pixels in the Phase II of \emph{VideoDP}, we measure the difference between all the pixels' RGB values in $V$ and $O$ using the expectation of mean squared error (MSE):

\begin{equation}
\small
E(MSE)=\frac{1}{|V|}\sum_{\forall (a,b,t)}\big(E[\theta(a,b,t)]-E[\hat{\theta}(a,b,t)]\big)^2
\end{equation}

where $\theta(a,b,t)$ and $\hat{\theta}(a,b,t)$ represent the RGB for pixel with coordinates $(a,b)$ and frame $t$ in $V$ and $O$, and $E(\cdot)$ denotes the expectation. The 3-dimensional RGBs are generally converted to gray for measuring the MSE \cite{MSE}, which can also be normalized to values in $[0,1]$.

Specifically, we conducted experiments to test how privacy budget $\epsilon$ influences the utility of the output videos. We fix $k$ (not optimal) for all the visual elements in two videos, and vary the privacy budget $\epsilon$ in the range from $0.8$ to $2.8$. Figure \ref{fig:kle1} and \ref{fig:kle2} present the KL-divergence values (where for all visual elements $k_j, j\in[1,n]$ is fixed as $6$ and $10$, respectively). Since KL values vary on the count distribution of sampled RGBs, we can observe that the KL values have different trends for different videos (PED and VEH) regardless of how $\epsilon$ varies. For instance, in video PED, the KL value first increases and then decreases as $\epsilon$ grows while it monotonically increases and then converges as $\epsilon$ increases in video VEH (since different videos may have very different RGB histograms).

In addition, we also evaluated the normalized MSE of the output videos (after Phase I, and after Phase II). Figure \ref{fig:mseebefore} and \ref{fig:mseeafter} show that the MSE (of the entire video) declines as $\epsilon$ increases. This matches the fact that larger $\epsilon$ trades off more privacy for better utility. Also, the normalized MSE has been greatly reduced after Phase II -- comparing the results in Figure \ref{fig:mseebefore} and \ref{fig:mseeafter}, which greatly improves the accuracy of the queries for video analyses. We also test how the parameter $k_j$ affects the output utility, and the optimal $k_j$ (based on minimum MSEs) is shown in Appendix \ref{sec:additional}.

\subsection{Evaluating Video Queries/Analysis}

We now evaluate the utility of private queries for video analyses. It is worth noting that the utility-driven private video randomly generated in \emph{VideoDP} can function any analysis, such as head counting, crowd density estimation and traffic flow analysis \cite{car, crowd-density, MOT16} in our experimental videos. We examine some representative queries for video analysis and compared the results with PINQ-based \cite{McSherry09} video query/analysis.

\begin{figure}[!h]
	\centering
	\subfigure[Precision vs $\epsilon$]{
		\includegraphics[angle=0, width=0.49\linewidth]{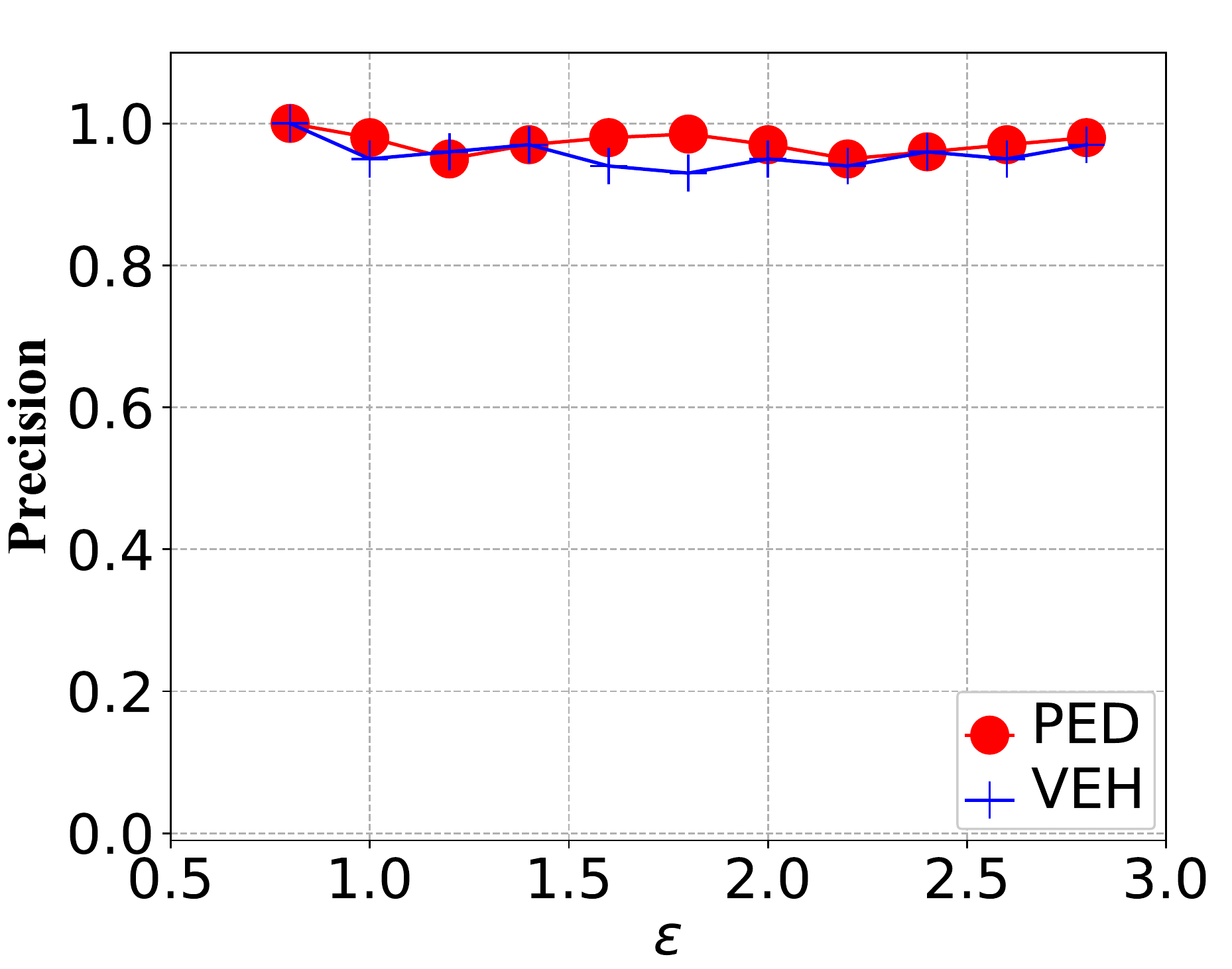}
		\label{fig:pre}}
	\hspace{-0.2in}
	\subfigure[Recall vs $\epsilon$]{
		\includegraphics[angle=0, width=0.49\linewidth]{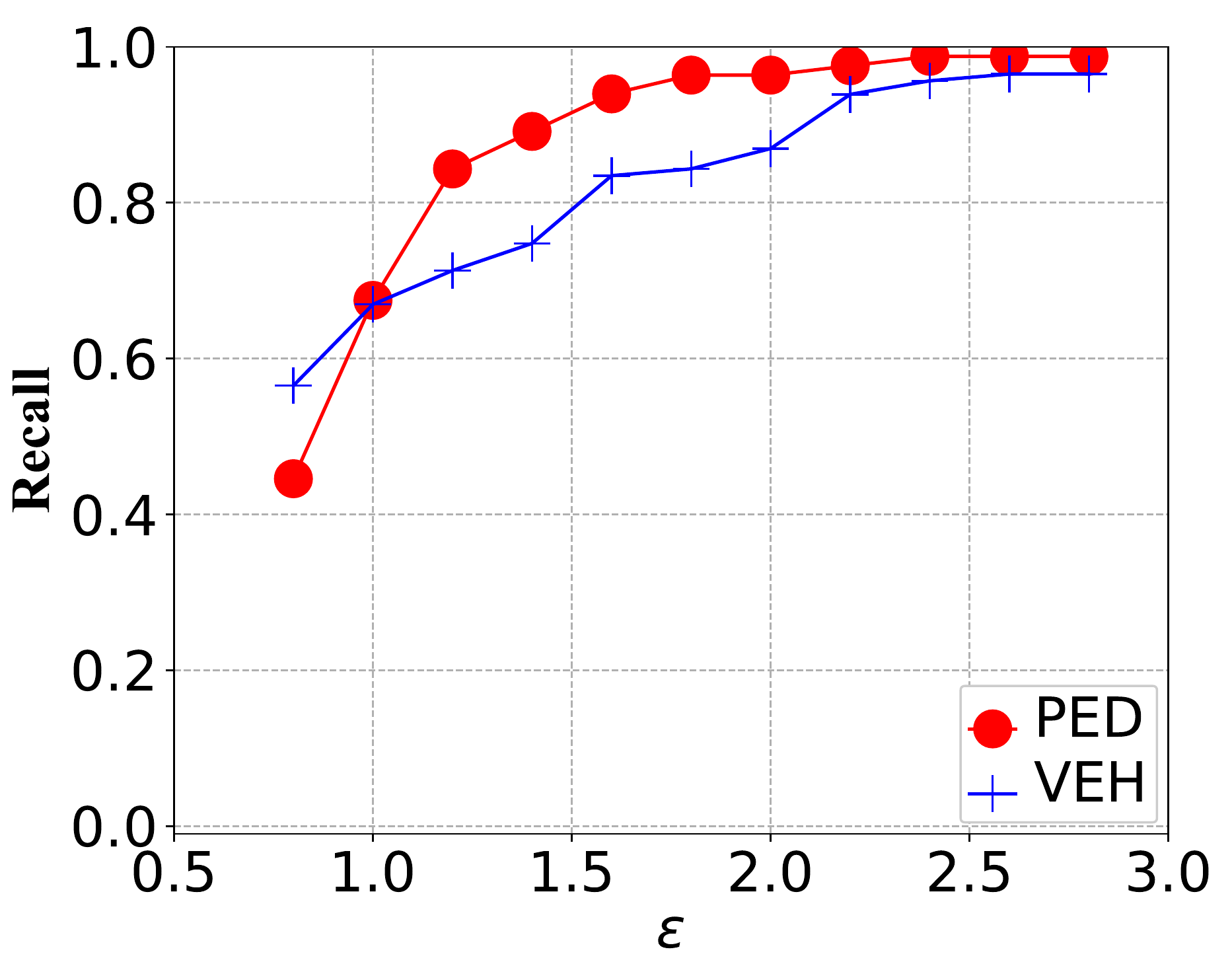}
		\label{fig:rec} }
	\caption[Optional caption for list of figures]
	{Visual Elements Detection and Tracking}
	\vspace{-0.1in}
	\label{fig:rp}
\end{figure}
\begin{figure*}[!tbh]
	\centering
	\subfigure[PED (Original)]{
		\includegraphics[angle=0, width=0.33\linewidth]{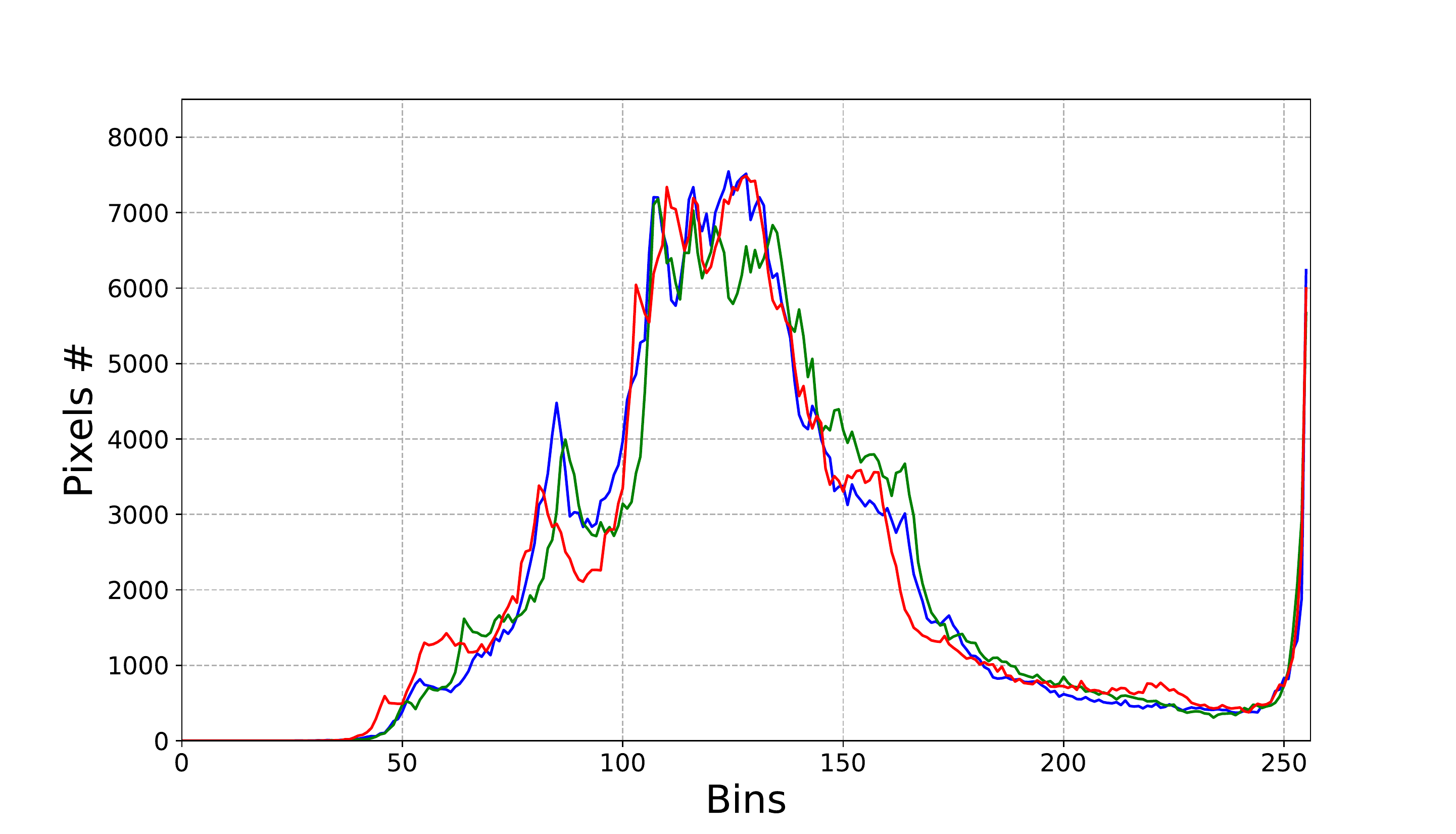}
		\label{fig:PED_RGB_Or}}
	\hspace{-0.2in}
	\subfigure[PED PINQ ($\epsilon=0.8$)]{
		\includegraphics[angle=0, width=0.33\linewidth]{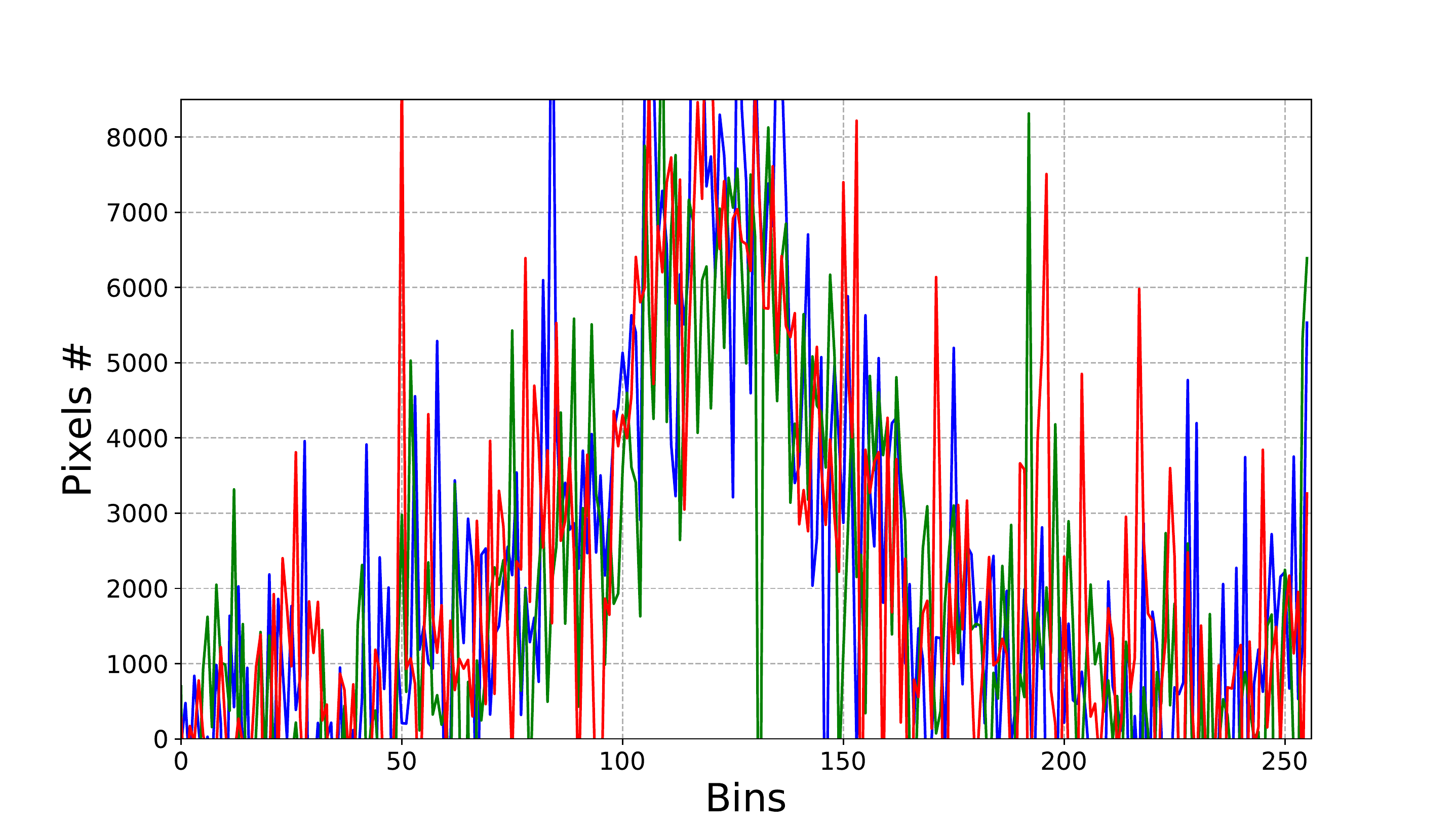}
		\label{fig:PED_RGB_PINQ} }
	\hspace{-0.2in}
	\subfigure[PED \emph{VideoDP} ($\epsilon=0.8$)]{
		\includegraphics[angle=0, width=0.33\linewidth]{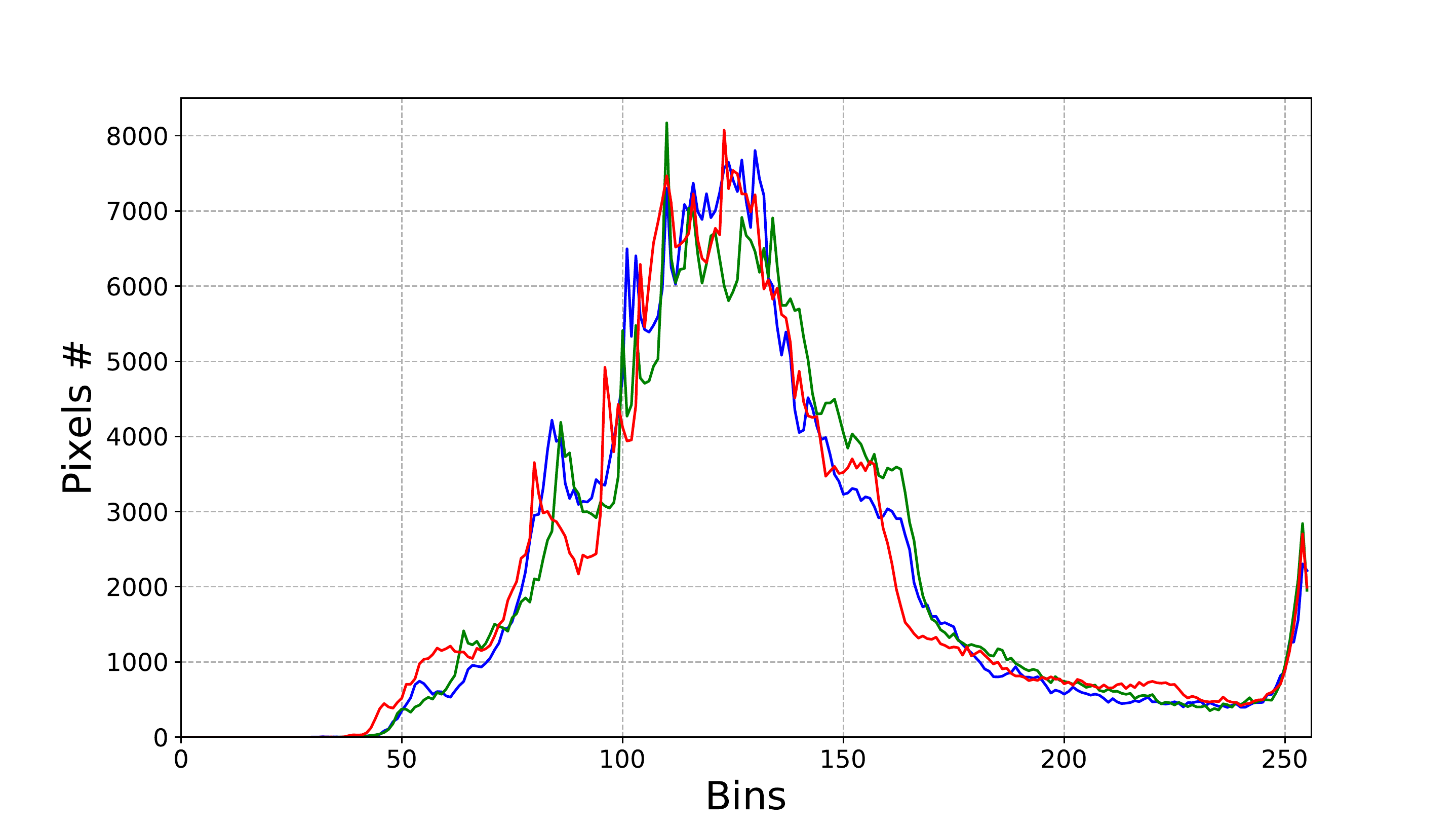}
		\label{fig:PED_RGB_N08} }
		
	\subfigure[VEH (Original)]{
		\includegraphics[angle=0, width=0.33\linewidth]{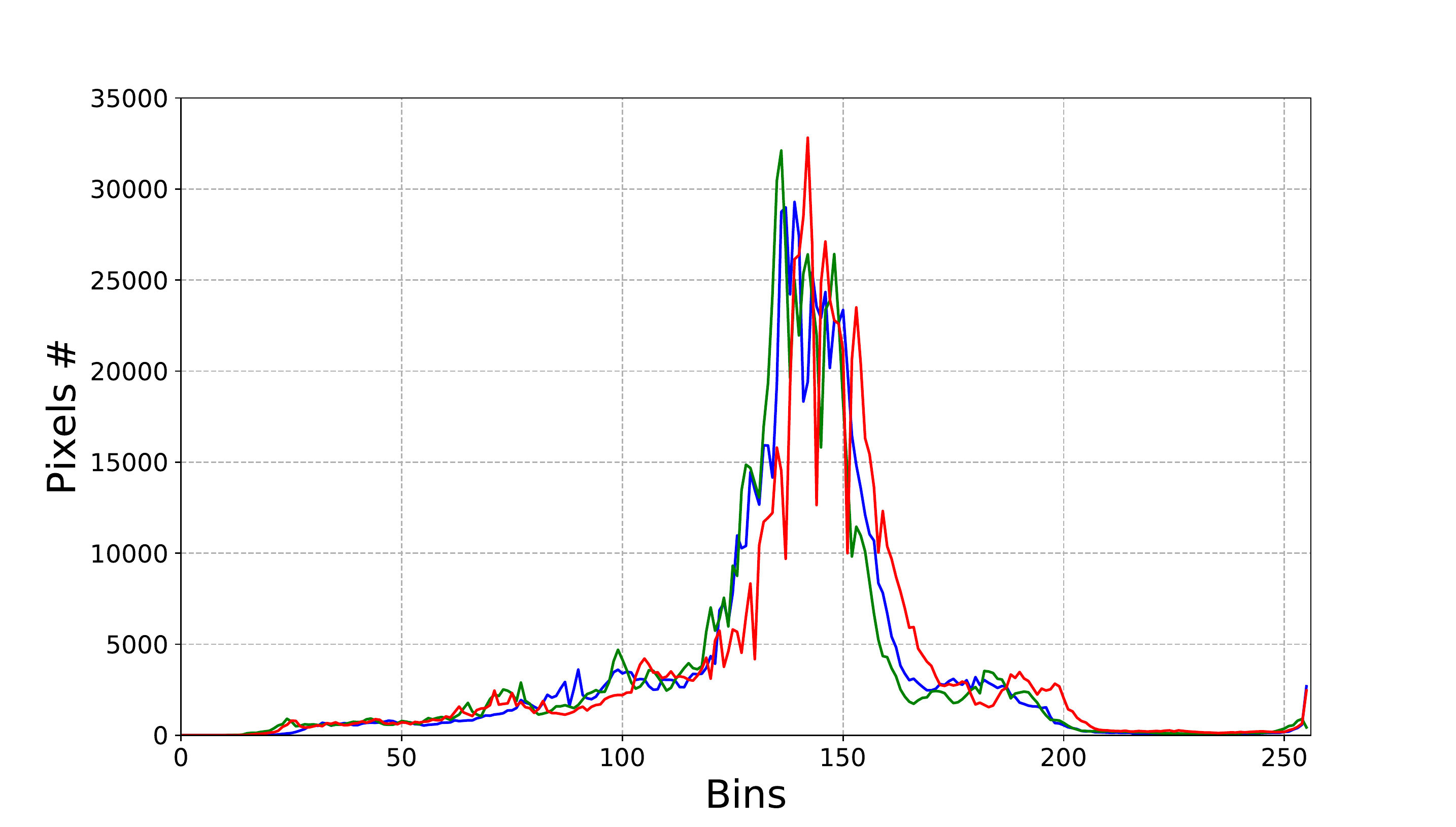}
		\label{fig:VEH_RGB_Or}}
	\hspace{-0.2in}
	\subfigure[VEH PINQ ($\epsilon=0.8$)]{
		\includegraphics[angle=0, width=0.33\linewidth]{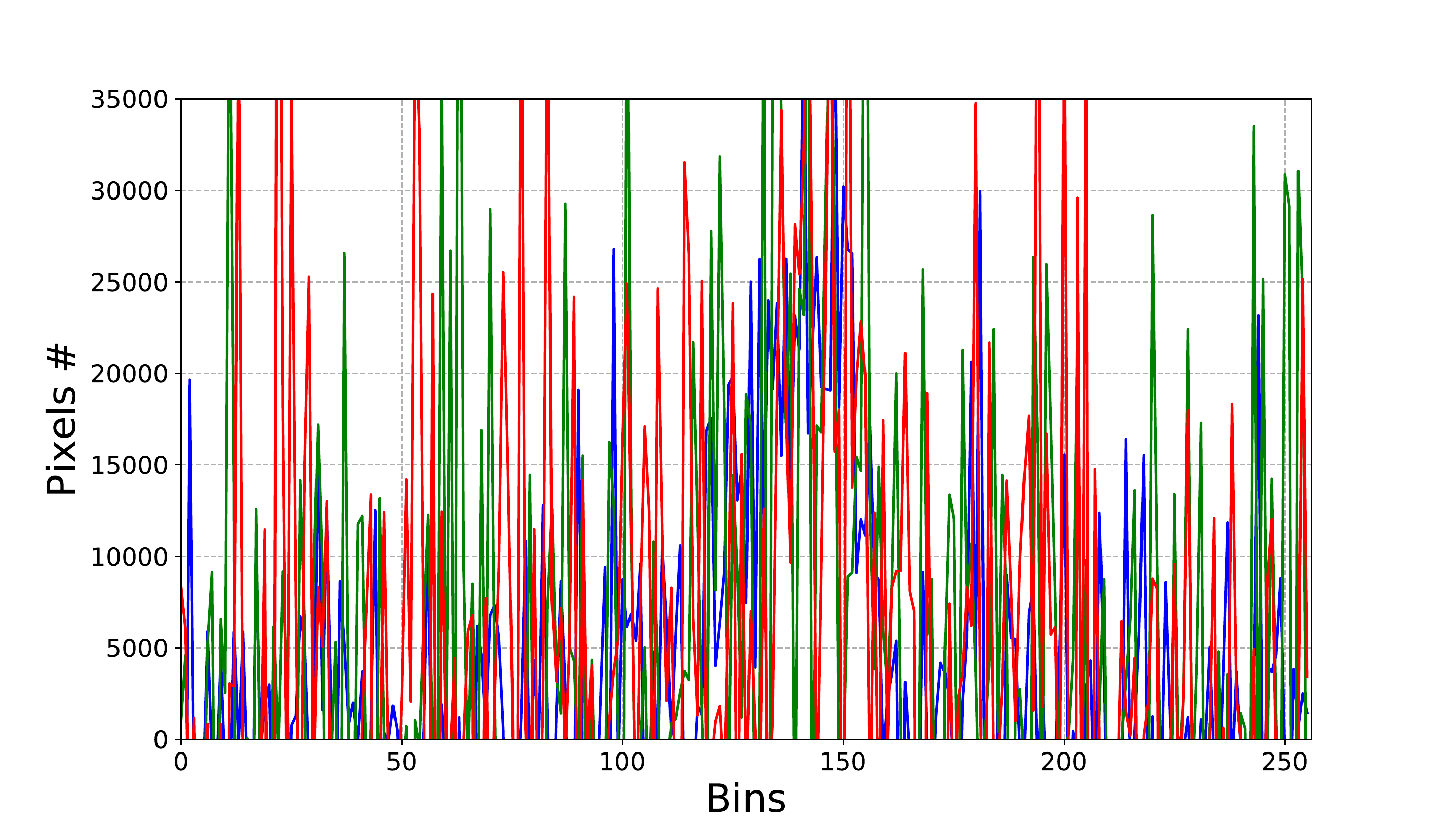}
		\label{fig:VEH_RGB_PINQ} }
	\hspace{-0.2in}
	\subfigure[VEH \emph{VideoDP} ($\epsilon=0.8$)]{
		\includegraphics[angle=0, width=0.33\linewidth]{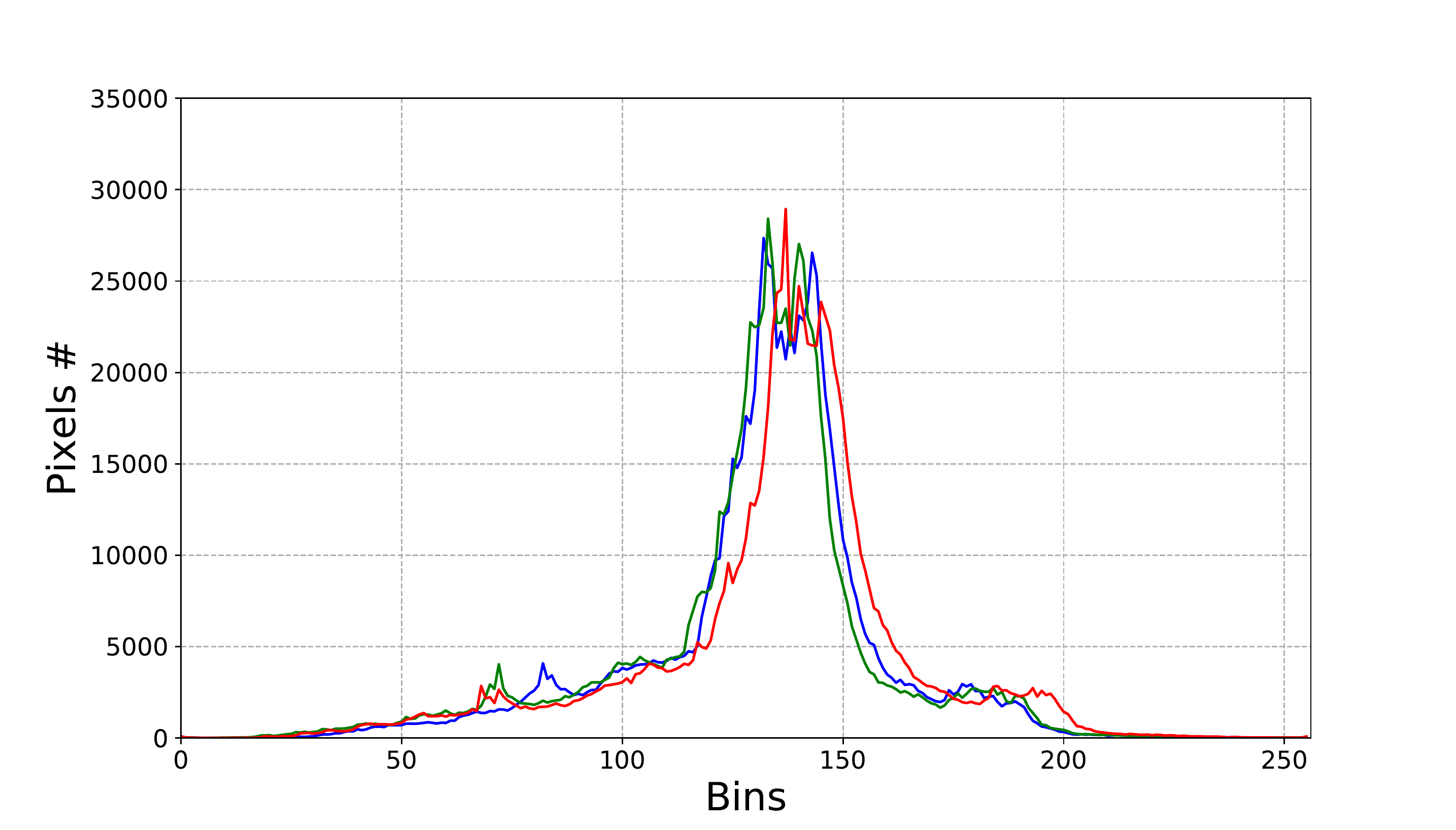}
		\label{fig:VEH_RGB_N08} }
	\caption[Optional caption for list of figures]
	{RGB Histograms in PED and VEH}
	\vspace{-0.1in}
	\label{fig:RGB}
\end{figure*}

\subsubsection{Pedestrian/Vehicle Detection}
In our first set of experiments, we have evaluated the detection/tracking accuracy. Considering the detection/tracking results applied to two original videos as the benchmark, we test the precision and recall of detecting/tracking humans and objects in the outputs. Specifically, \emph{precision} returns the percent of true pedestrians/vehicles out of all the detected/tracked results in the videos. \emph{Recall} returns the percent of detected/tracked true pedestrians/vehicles out of all the true pedestrians/vehicles (the benchmarking results obtained from the original videos). Figure \ref{fig:rp} demonstrates the precision and recall on varying privacy budget $\epsilon$. The precision can always be maintained with a high accuracy (close to $1$). The recall grows quickly as $\epsilon$ increases (since larger $\epsilon$ can generate more accurate random private videos).

\subsubsection{RGB Histograms}
We next compare \emph{VideoDP} and PINQ \cite{McSherry09} using representative video queries/analysis at the pixel and visual elements levels. Pixel level analysis plays an important role in different applications (e.g., frame classification)\cite{imghistogram}. If queries are applied to pixels in a frame or video, the sensitivity might be extremely large using PINQ \cite{McSherry09} since adding or removing a visual element may lead to thousands of changed pixels (w.r.t. the visual element). For instance, an important pixel level query requests ``the RGB histogram of some frames in a video'' \cite{RGB_His}, the sensitivity should be the \emph{maximum number of pixels} involved in different visual elements in the frames. Figure \ref{fig:RGB} demonstrates the RGB histograms for a representative frame in video PED and VEH.
We set the privacy budget $\epsilon$ as $0.8$, and the sensitivity in PINQ for video PED and VEH would be greater than 1000. Then, the PINQ results are more fluctuated and significantly deviated as shown in Figure \ref{fig:PED_RGB_PINQ} and \ref{fig:VEH_RGB_PINQ}. We can also observe that results of utility-driven private video lie closer to the original results even if privacy budget $\epsilon$ is small (see Figure \ref{fig:PED_RGB_N08} and \ref{fig:VEH_RGB_N08}). 

\subsubsection{Pedestrian/Vehicle Stay Time}

\begin{figure*}[!tbh]
	\centering
	\subfigure[Pedestrian Stay Time (PINQ)]{
		\includegraphics[angle=0, width=0.5\linewidth]{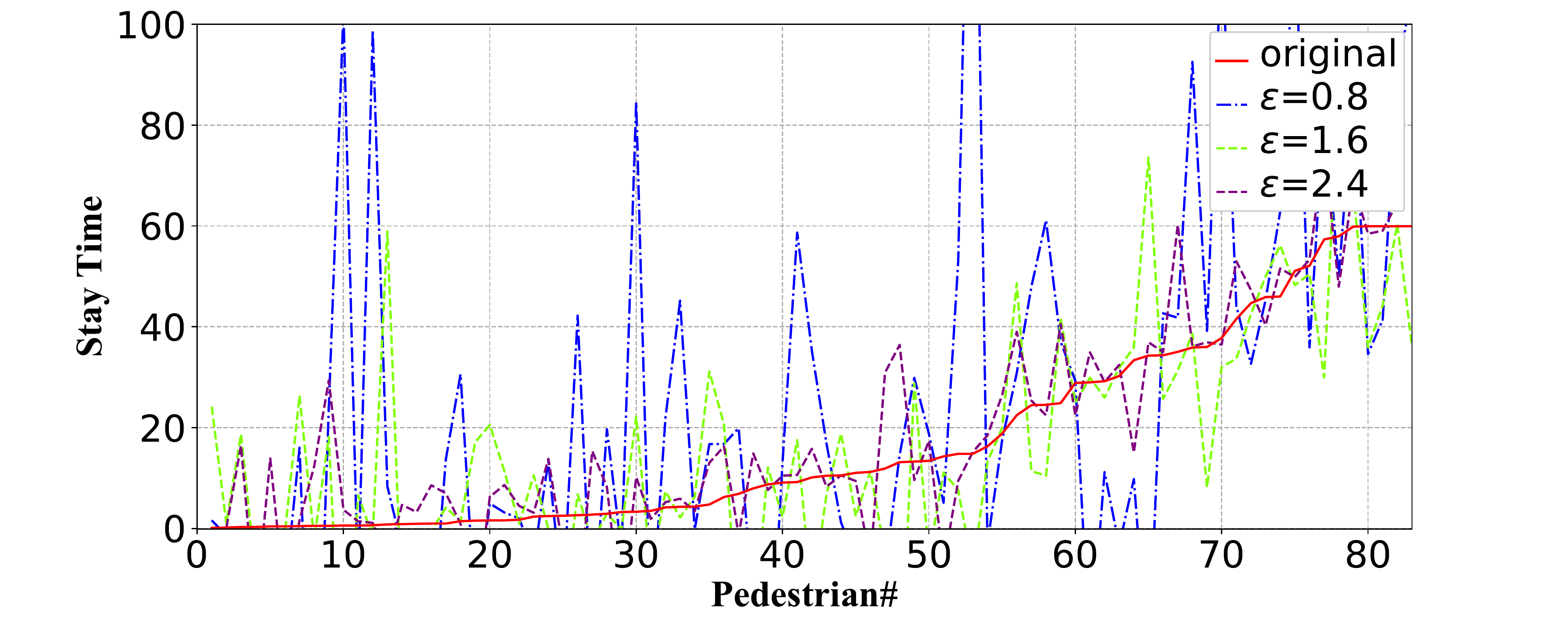}
		\label{fig:PED_FRAME_PINQ} }
		\hspace{-0.3in}
	\subfigure[Pedestrian Stay Time (\emph{VideoDP})]{
		\includegraphics[angle=0, width=0.5\linewidth]{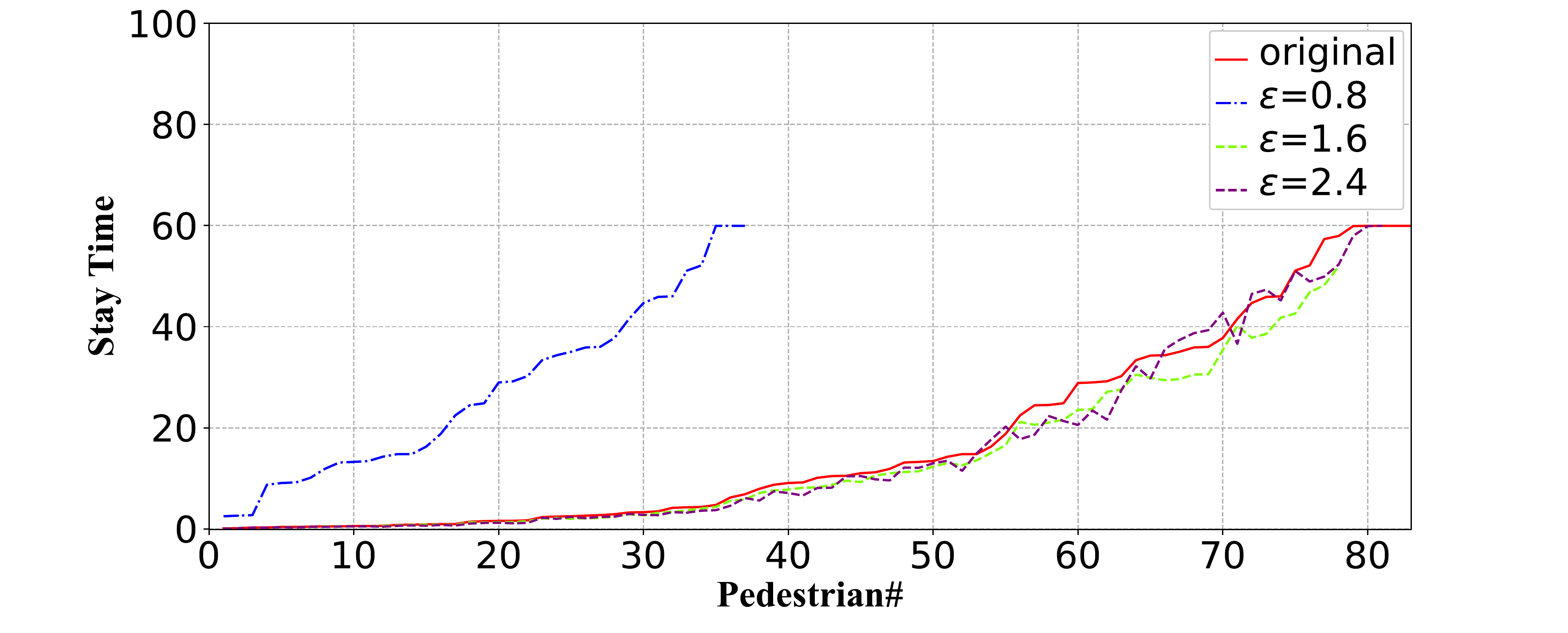}
		\label{fig:PED_FRAME_VDP}}
	\vspace{-0.05in}
	\caption[Optional caption for list of figures]
	{Pedestrian Stay Time}
	\label{fig:PED_FRAME}
\end{figure*}
\begin{figure*}[!tbh]
	\centering
	\subfigure[Moving Upstream (PINQ)]{
		\includegraphics[angle=0, width=0.5\linewidth]{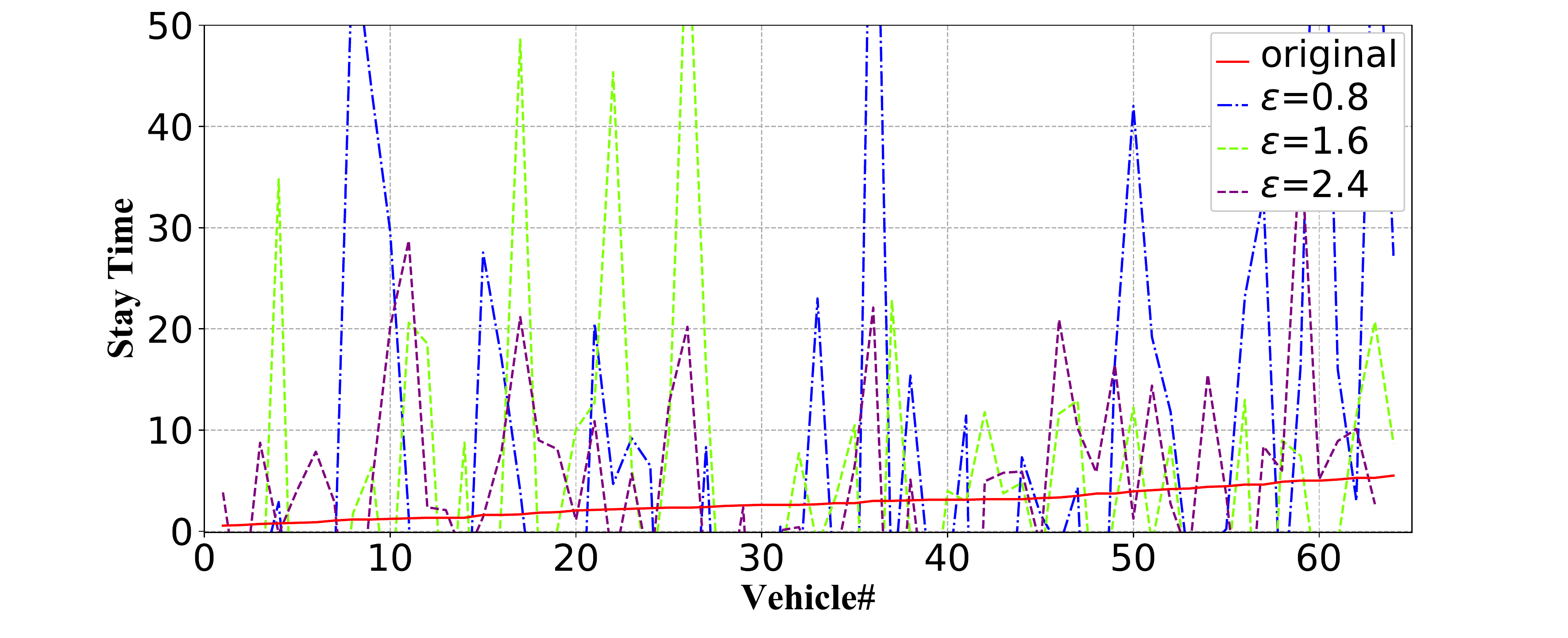}
		\label{fig:VEH_UP_PINQ} }
		\hspace{-0.3in}
	\subfigure[Moving Upstream (\emph{VideoDP})]{
		\includegraphics[angle=0, width=0.5\linewidth]{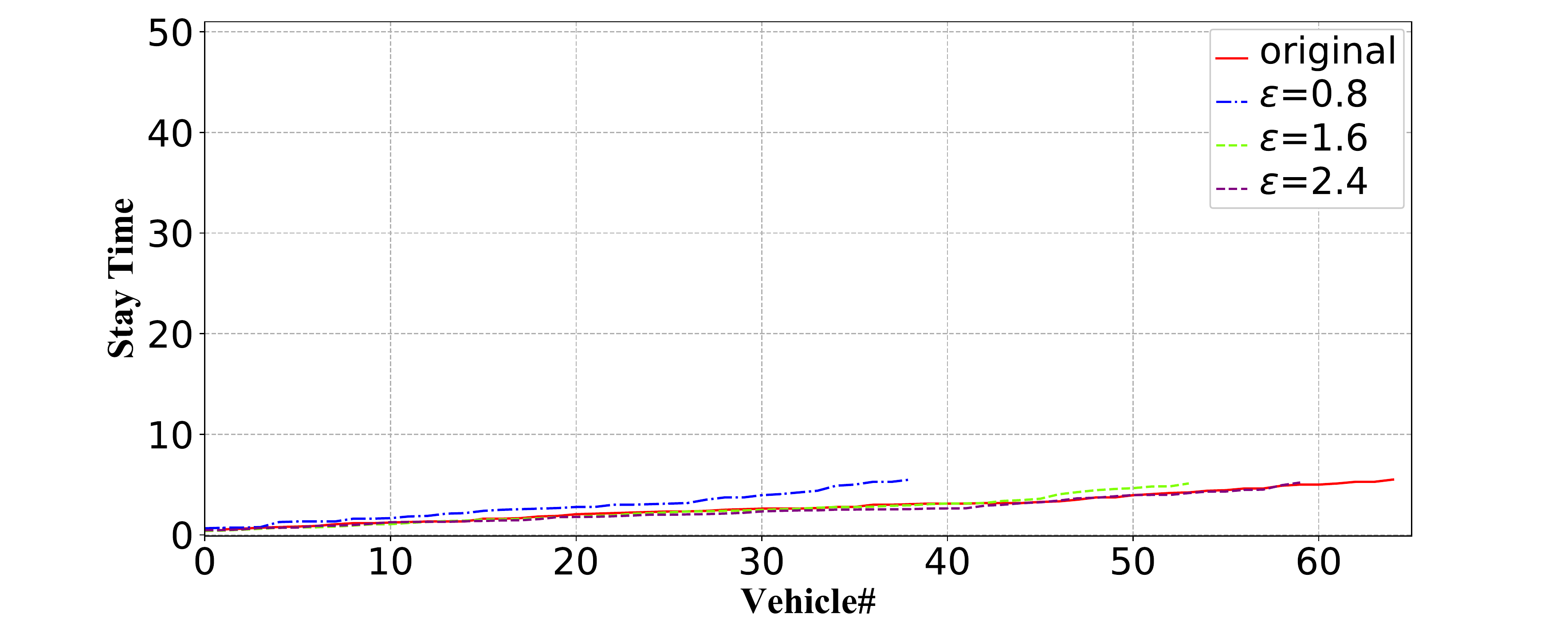}
		\label{fig:VEH_UP_VDP}}
	\subfigure[Moving Downstream (PINQ)]{
		\includegraphics[angle=0, width=0.5\linewidth]{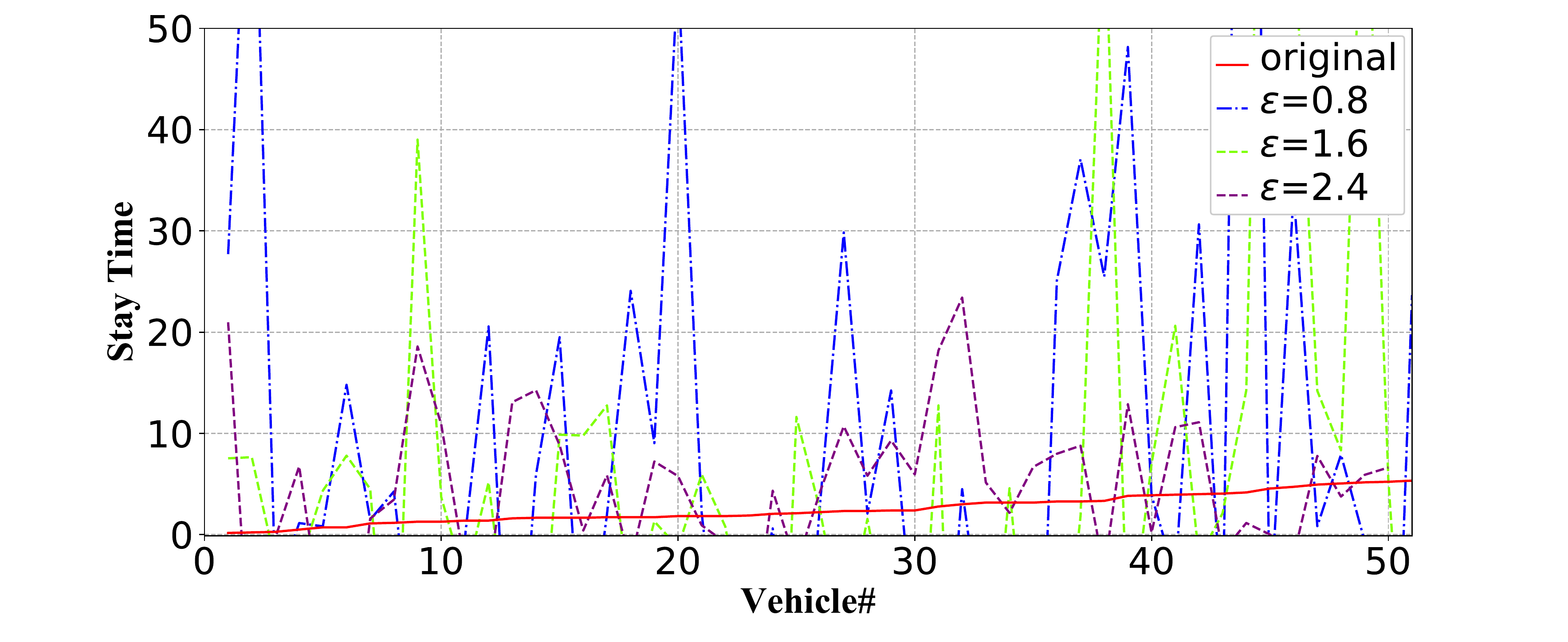}
		\label{fig:VEH_DOWN_PINQ} }
		\hspace{-0.3in}
	\subfigure[Moving Downstream (\emph{VideoDP})]{
		\includegraphics[angle=0, width=0.5\linewidth]{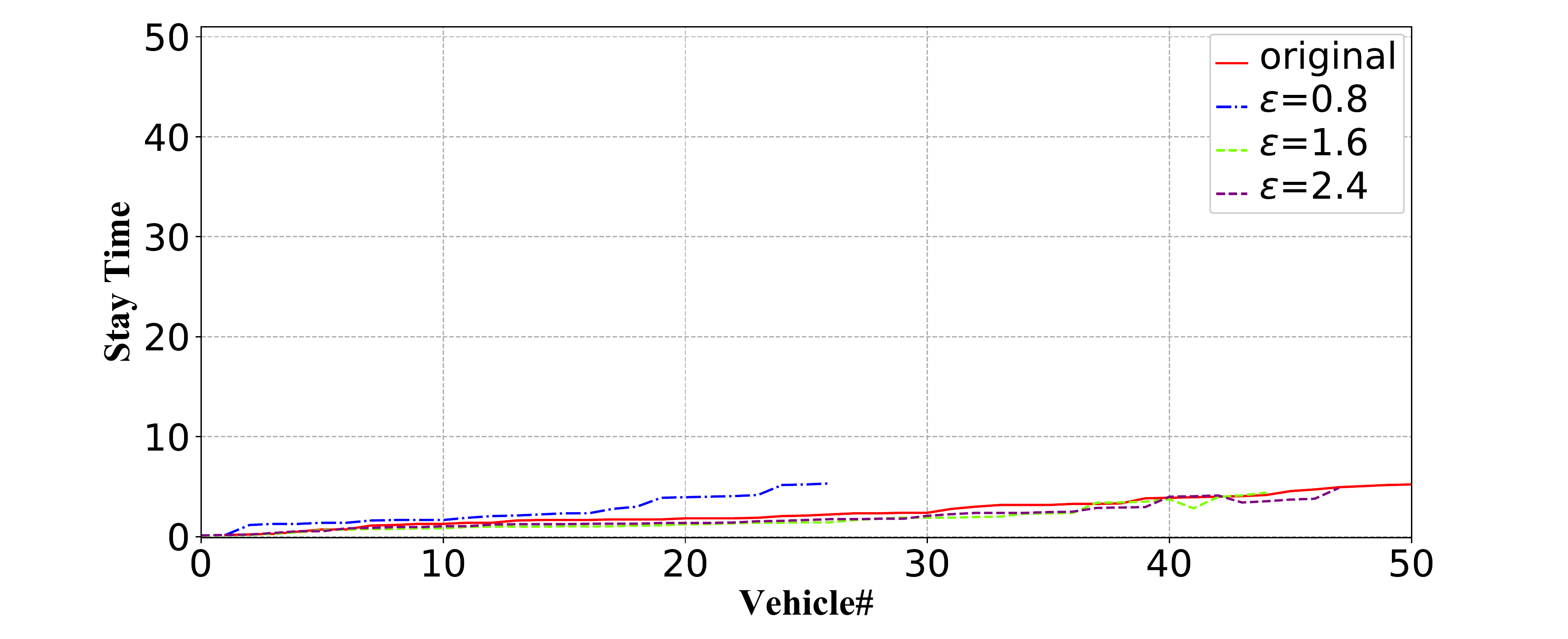}
		\label{fig:VEH_DOWN_VDP}}
	\caption[Optional caption for list of figures]
	{Vehicle Stay Time}\vspace{-0.1in}
	\label{fig:VEH_FRAME}
\end{figure*}

Besides the pixel level queries, \emph{VideoDP} can also privately return query results for detected visual elements in different applications. For instance, a query returns ``how long each pedestrian/vehicle stays in the video'' (denoted as stay time or number of stay frames). Then, pedestrians/vehicles are detected/tracked in all the frames, and then query results can be computed/aggregated and returned for private analysis.

\vspace{0.05in}

\noindent\textbf{(1) Pedestrians}. 
In the PED video, 83 pedestrians are walking on the street. How long each pedestrian stays in the video can be utilized to learn the human behavior. Figure \ref{fig:PED_FRAME} presents the original results, PINQ results and \emph{VideoDP} results for the PED video where three different privacy budgets $\epsilon$ (0.8, 1.6 and 2.4) are specified. The 83 pedestrians in the PED video (marked on the x axis), and the stay time is ranked from short to long (see the red curve in two subfigures). In PINQ (Figure \ref{fig:PED_FRAME_PINQ}), the stay times of all the pedestrians are overly obfuscated even if $\epsilon$ is large since sensitivity $\Delta$ should be set as 60 (for even longer videos, $\Delta$ should be larger). 
Nevertheless, \emph{VideoDP} significantly outperforms PINQ for such query/analysis.
As shown in Figure \ref{fig:PED_FRAME_VDP},  in case of $\epsilon=0.8$ (small privacy budget), approximately $40$ distinct pedestrians are detected in the utility driven video. Although not all the pedestrians are sampled in \emph{VideoDP}, the distribution of all the stay times (of all the sampled pedestrians) still lies close to the original result. As $\epsilon$ increases to 1.6, the query results obtained from \emph{VideoDP} converges to the original results (however, PINQ results are still fluctuated).

\vspace{0.05in}

\noindent\textbf{(2) Vehicle}. In the VEH video, there are 115 distinct vehicles driving on the highway. We define the two-way moving directions as ``upstream'' and ``downstream'', respectively. Figure \ref{fig:VEH_FRAME} demonstrates the length of time the vehicles stay in the video (upstream and downstream), which can be utilized to estimate the moving speed of vehicles, queue length estimation, etc. We can draw similar observations for the stay times of vehicles for both moving downstream direction and upstream direction in the VEH video as the pedestrians in the PED video. \emph{VideoDP} significantly outperforms PINQ which requests a high sensitivity in such query. 
\subsubsection{Vehicle Density}
We also conduct experiments to compare \emph{VideoDP} and PINQ on queries which request a smaller sensitivity. For instance, the vehicle density query over the video returns the vehicle count in each frame of the video (sensitivity $\Delta=1$), which can also facilitate the analyst to learn the traffic flow. 

Figure \ref{fig:COUNT_FRAME} demonstrates the count of detected vehicles in each frame, including the original results, PINQ results (Figure \ref{fig:VEH_FRAME_PINQ}) and \emph{VideoDP} results (Figure \ref{fig:VEH_FRAME_VDP}) where three different privacy budgets $\epsilon$ (0.8, 1.6 and 2.4) are also specified. Note that every vehicle only appears in a few frames of the video (this also occurs in longer videos). The noise results are both acceptable in PINQ ($\Delta=1$) and \emph{VideoDP}. However, the counts of vehicles are more fluctuated in PINQ as $\epsilon$ is small. 
\subsubsection{Complex Analysis}

Since any video analysis algorithm can be broken down into queries, complex analysis such as RGB histogram analysis \cite{RGB_His} and deep learning \cite{emotion_recognition} can also be directly applied to the (random) utility-driven private video while ensuring $\epsilon$-DP for \emph{VideoDP}. At this time, PINQ-based differentially private scheme should be redesigned for each analysis with appropriate privacy budget allocation and composition analysis. Indeed, our \emph{VideoDP} can address such limitation of PINQ with better \emph{flexibility}. Due to space limit, we leave the experimental studies for different complex analyses as our future work.  

\begin{figure*}[!tbh]
	\centering
	\subfigure[Vehicle Count in each Frame (PINQ)]{
		\includegraphics[angle=0, width=0.48\linewidth]{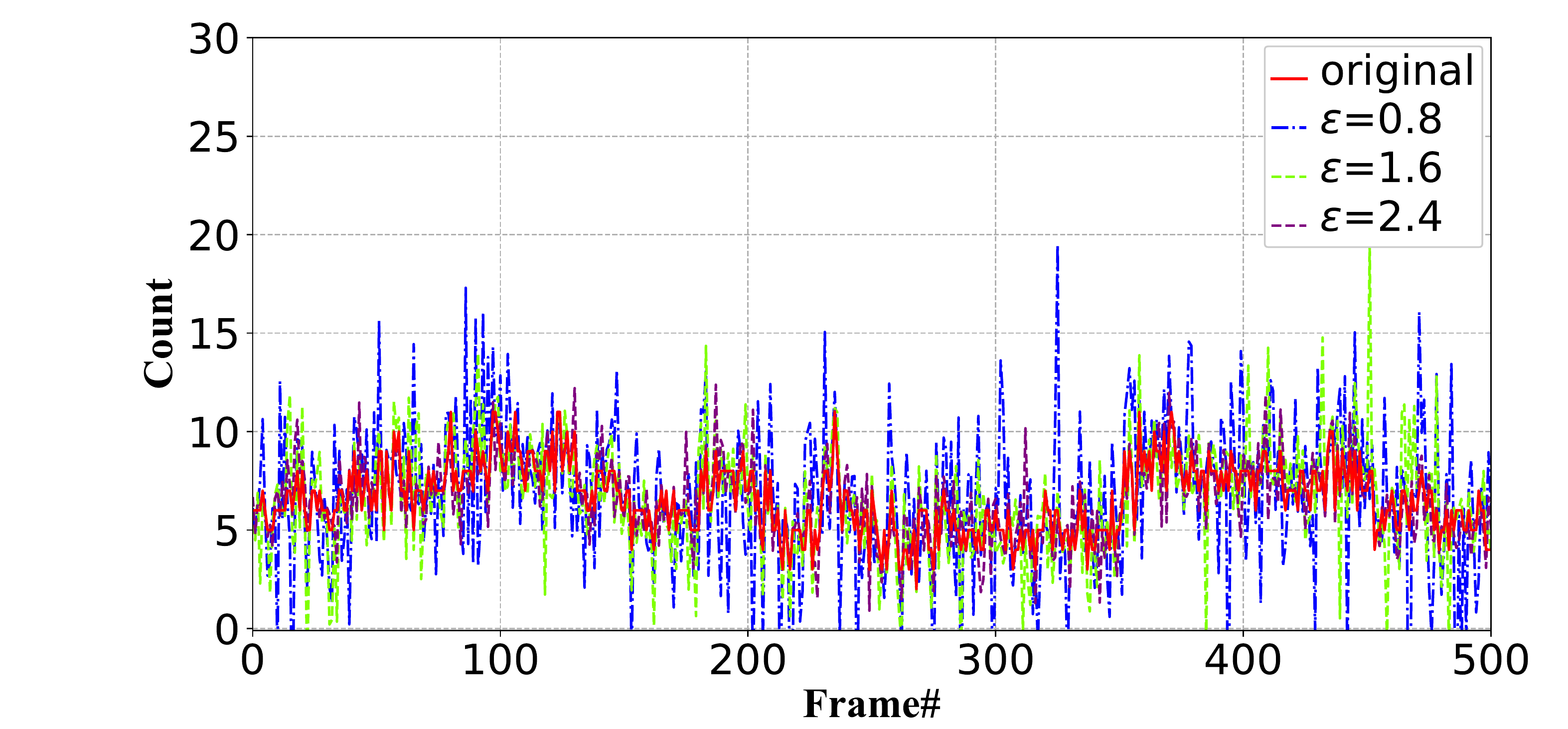}
		\label{fig:VEH_FRAME_PINQ} }
		\hspace{-0.2in}
	\subfigure[Vehicle Count in each Frame (\emph{VideoDP})]{
		\includegraphics[angle=0, width=0.48\linewidth]{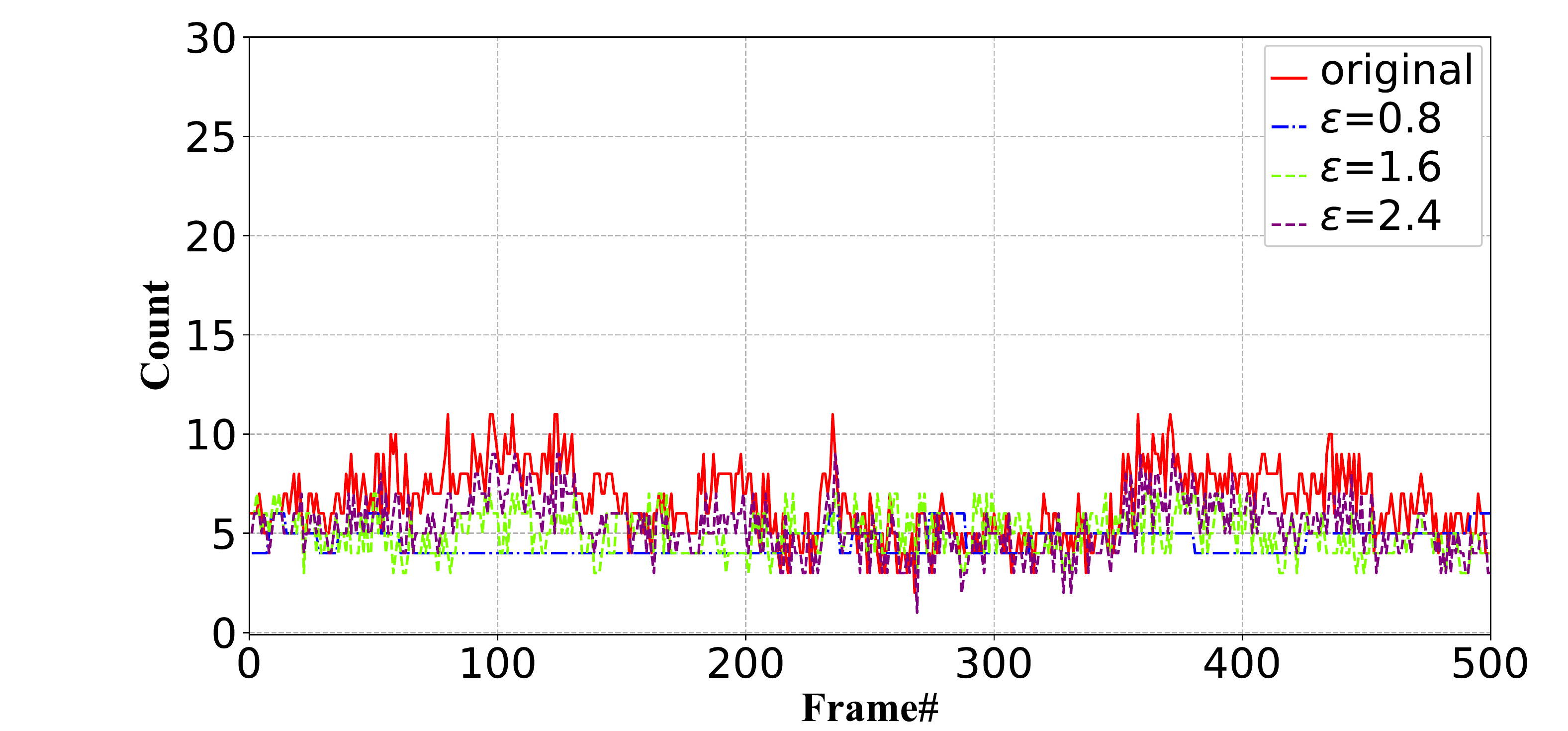}
		\label{fig:VEH_FRAME_VDP}}
	\vspace{-0.05in}
	\caption[Optional caption for list of figures]
	{Vehicle Count in Each Frame}
	\label{fig:COUNT_FRAME}
\end{figure*}

\subsection{Computational Performance}
We have also evaluated the computational performance of our \emph{VideoDP} (especailly Phase I and II). To generate a utility-driven private video for a high-resolution video (e.g., video PED includes $1920\times 1080$ pixels in each frame), it takes around 20 minutes to detect and track all the visual elements, sample pixels (including optimizing each $k_j$), and interpolate the video. Similarly, it takes around 12 minutes to generate a synthetic video for video VEH ($1280\times720$ pixels in each frame) by executing all the procedures in \emph{VideoDP}. Such one-time computation is acceptable to process billions of pixels (offline cost), and then multiple queries/analysis can be directly applied to the private video.

For longer videos, we can split the input video into multiple fragments (e.g., 1 minute per fragment). Then, we can still apply \emph{VideoDP} to efficiently generate the utility-driven private videos for all the fragments which are integrated later. In many videos (e.g., traffic monitoring videos), VEs move rapidly and appear in the video for a few seconds. Then, fragmentation, generation and integration would not affect the performance (\emph{if each visual element can be protected in a fragment and do not appear in multiple fragments, it can also be protected in the integrated video}). In case that certain VEs stay in a video for a long time, considering them as different VEs in different fragments to protect may change the privacy protection in the integrated output video. To address this in such VEs, we can split the privacy budget w.r.t. each VE to shares for different segments (sequential composition \cite{McSherry09}). This may lower the probability of retaining such VEs in the output, but it could greatly improve the efficiency and scalability.

Finally, the runtime for deriving the private query/analysis results from the utility-driven private video for analysts (Phase III) in \emph{VideoDP} is quite close to querying/analyzing the original video (which is expected to be more efficient than PINQ).

\section{Related Work}
\label{sec:related}
Dwork~\cite{Dwork06differentialprivacy} first proposed the notion of differential privacy that provides rigorous privacy protection regardless of background knowledge. In literature, such rigorous privacy notion has been extended to sanitize and release data in statistical databases \cite{Dwork06}, histograms \cite{Cormode12}, location data \cite{ninghui13}, search logs \cite{KorolovaKMN09}, and graphs \cite{HayLMJ09}, etc. To our best knowledge, we take the first step to address the deficiency in differentially private video analysis. In particular, sampling based randomization mechanisms have been proposed to generate probabilistic true data from the input datasets while ensuring differential privacy (e.g., microdata \cite{2018safepub, Ning2012sampling}, search logs \cite{Hong2014collaborative}). Our utility-driven private video is randomly generated via pixel sampling, and universally supports any video analysis with differential privacy.

Since \emph{VideoDP} locally perturbs the input video by the video owner (and then flexibly offers queries/analysis to untrusted analysts), 
the emerging local differential privacy (LDP) models \cite{CormodeLDP18,rappor14,histogram1} are also relevant to this work. The state-of-the-art LDP models perturb local data by the owners to generate statistics for histograms/heavy hitters \cite{rappor14,histogram1,ldpusenix17,zhan16}, reconstruct social graphs \cite{zhan17}, and function frequent itemset mining \cite{ldpfrequent18}.

Most of the existing video sanitization techniques use straightforward measures for quantifying the privacy loss in videos. For instance, in \cite{Fidaleo:privacy,Moncrieff:privacy}, if faces are present, then it is considered as complete privacy loss, otherwise no privacy loss. Besides only considering the faces as privacy leakage, recent works also investigated the privacy leakage in the activities, visited places and other implicit channels (e.g., when and where the video is recorded) \cite{SainiMTAP_W3}. Fan \cite{fandbsec18} applied Laplace noise to obfuscate pixels in an image to ensure differential privacy for protecting specific regions of an image. However, the image quality has been significantly reduced. Both the privacy notion and the Laplace noise (high sensitivity) cannot be applied to pixels in videos. 

Moreover, most existing techniques directly adopted computer vision techniques \cite{Fidaleo:privacy,Koshimizu:privacy} to first \emph{detect} faces and/or other sensitive regions in the video frames and then \emph{obscure} them. 
However, such \textit{detect-and-protect} solutions have some limitations. 
For instance, the \textit{detect-and-protect} techniques cannot formally quantify and bound the privacy leakage. Thus, the video owner does not know how much risks any individual can be identified from the video.  
In addition, blurred regions might still be reconstructed by deep learning methods \cite{dl1,dl2}. Our \emph{VideoDP} can address these limitations with strong privacy protection against arbitrary prior knowledge.

\section{Conclusion}
\label{sec:concl}

In this paper, to the best of our knowledge, we take the first step to study the problem of video analysis with \emph{differential privacy guarantee}. Specifically, we have proposed a new sampling based differentially private mechanism to randomly generate utility-driven private videos. \emph{VideoDP} has also provided a universal platform for analysts to privately conduct any query/analysis over the (random) utility-driven private video with differential privacy guarantee. We have conducted extensive experiments to validate that the performance of \emph{VideoDP} significantly outperforms PINQ-based video analyses in many different applications. 


\begin{appendices}

\section{Optimal $k_j$ for VE $\Upsilon_j$}
\label{sec:opt}

\subsection{Equations for Different Pixels}
\label{sec:border}

If pixel $(a,b,t)$ is a non-border pixel, we have Equation \ref{eq:mid} to represent the relation between the RGB expectation of any pixel $(a,b,t)$ and the RGB expectation of its four neighbors (denoted as $\hat{\theta}_N, \hat{\theta}_S, \hat{\theta}_W$ and $\hat{\theta}_E$). We now briefly discuss how to derive such relation. 

First, if pixel $(a,b,t)$ is sampled, then the RGB expectation equals $Pr(a,b,t)*\theta(a,b,t)$ where $Pr(a,b,t)$ is the probability of sampling $(a,b,t)$ and $\theta(a,b,t)$ denotes its RGB in the original video $V$. 

Second, if pixel $(a,b,t)$ is not sampled, then it will be interpolated based on the RGBs of its neighbors. There are five subcases (denoting the probabilities that $(a,b,t)$ has 0, 1, 2, 3 and 4 neighbors before interpolation as $\sigma_0(a,b,t), \sigma_1(a,b,t)$, $\sigma_2(a,b,t)$, $\sigma_3(a,b,t), \sigma_4(a,b,t)$, respectively): 

\begin{enumerate}
\item 0 neighbor: all its neighbors are not sampled in Phase I. Then, the probability share is $\sigma_0(a,b,t)*0$. 

\item 1 neighbor: 3 of its neighbors are not sampled in Phase I. Then, the probability share is:

\scriptsize
\begin{equation}
\small
\sigma_1(a,b,t)*[1-Pr(a,b,t)]*\frac{E(\hat{\theta}_N)+E(\hat{\theta}_S)+E(\hat{\theta}_W)+E(\hat{\theta}_E)}{4}
\end{equation}
\normalsize

where all 4 neighbors can be the one used for interpolation.

\item 2 neighbors: 2 of its neighbors not sampled in Phase I. Then, the probability share is:

\scriptsize
\begin{equation}
\sigma_2(a,b,t)*[1-Pr(a,b,t)]*\frac{3E(\hat{\theta}_N)+3E(\hat{\theta}_S)+3E(\hat{\theta}_W)+3E(\hat{\theta}_E)}{6*2}
\end{equation}
\normalsize

where 6 different combinations of two neighbors can be used for interpolation and the interpolated RGB is the average of two neighbors' RGBs.

\item 3 neighbors: 1 of its neighbors is not sampled in Phase I. Then, the probability share is:

\scriptsize
\begin{equation}
\sigma_3(a,b,t)*[1-Pr(a,b,t)]*\frac{3E(\hat{\theta}_N)+3E(\hat{\theta}_S)+3E(\hat{\theta}_W)+3E(\hat{\theta}_E)}{4*2}
\end{equation}
\normalsize

where 4 different combinations of two neighbors can be used for interpolation and the interpolated RGB is the average of three neighbors' RGBs.

\item 4 neighbors: no neighbor is suppressed in sampling. Then, the probability share is:

\scriptsize
\begin{equation}
\sigma_4(a,b,t)*[1-Pr(a,b,t)]*\frac{E(\hat{\theta}_N)+E(\hat{\theta}_S)+E(\hat{\theta}_W)+E(\hat{\theta}_E)}{4}
\end{equation}
\normalsize

where only 1 combination of 4 neighbors can be used for interpolation and the interpolated RGB is the average of 4 neighbors' RGBs.

\end{enumerate}

Similarly, if pixel $(a,b,t)$ is on the border but not at the corner of the $t$th frame (w.l.o.g., the left border), then we have:

\scriptsize
\begin{align*}
\small
&E[\hat{\theta}(a,b,t)]
=Pr(a,b,t)*\theta(a,b,t)+\sigma_0(a,b,t)*0\nonumber\\
+&\frac{\sigma_1(a,b,t)*[1-Pr(a,b,t)]*[E(\hat{\theta}_N)+E(\hat{\theta}_S)+E(\hat{\theta}_E)]}{3}\nonumber\\
+&\frac{\sigma_2(a,b,t)*[1-Pr(a,b,t)]*[2E(\hat{\theta}_N)+2E(\hat{\theta}_S)+2E(\hat{\theta}_E)]]}{3*2}\nonumber
\end{align*}
\begin{align}
\small
+&\frac{\sigma_3(a,b,t)*[1-Pr(a,b,t)]*[E(\hat{\theta}_N)+E(\hat{\theta}_S)+E(\hat{\theta}_E)]}{3}
\label{eq:border}
\end{align}
\normalsize

If pixel $(a,b,t)$ is located at the corner of the $t$th frame (w.l.o.g., the upper-left corner), then we have:

\scriptsize
\begin{align}
\small
&E[\hat{\theta}(a,b,t)]
=Pr(a,b,t)*\theta(a,b,t)+\sigma_0(a,b,t)*0\nonumber\\
+&\frac{\sigma_1(a,b,t)*[1-Pr(a,b,t)]*[E(\hat{\theta}_S)+E(\hat{\theta}_E)]}{2}\nonumber\\
+&\frac{\sigma_2(a,b,t)*[1-Pr(a,b,t)]*[E(\hat{\theta}_S)+E(\hat{\theta}_E)]]}{2}
\label{eq:corner}
\end{align}
\normalsize

\subsection{Solving Algorithm}
\label{sec:solver}

The optimal number of distinct RGBs $k_j$ (to allocate privacy budget) is computed based on minimizing the MSE expectation of visual element $\Upsilon_j$ (averaged by the number of pixels). Thus, we solve the following optimization (which is equivalent to Equation \ref{eq:mse}): 

\begin{equation*}
    \argmin_{k_j}\sum_{\forall (a,b,t)\in \Upsilon_j}\big(E[\theta(a,b,t)]-E[\hat{\theta}(a,b,t)]\big)^2
\end{equation*}

We now present how to derive the optimal $k_j$ (given multi-scale RGB selection) and the corresponding RGB expectation of all the pixels in VE $\Upsilon_j$ (in all the frames). Per Equation \ref{eq:mid} (for non-border pixels), Equation \ref{eq:border} (for border-but-not-corner pixels) and Equation \ref{eq:corner} (for corner pixels), for pixel $|\Upsilon_j|$ pixels, we can have $|\Upsilon_j|$ equations for  $|\Upsilon_j|$ variables, each of which is the RGB expectation of a pixel. W.l.o.g., assuming that $\Upsilon_j$ has $A*B$ pixels in a rectangle with the coordinates $(1,1), (1,2), \dots, (1,B), \dots (A,1), \dots (A,B)$, thus we have (frame ID $t$ is skipped for simplicity of notations):

\begin{equation*}
\small
\left\{
        \begin{array}{lr}
 E[\hat{\theta}(1,1)]=Pr(1,1))*\theta(1,1)+[1-Pr(1,1)]*\\
 \quad\quad\quad\quad\quad \frac{(\sigma_1(1,1)+\sigma_2(1,1))*(E[\hat{\theta}(1,2)]+E[\hat{\theta}(2,1)])}{2}\\
 
E[\hat{\theta}(1,2)]=Pr(1,2)*\theta(1,2)+[1-Pr(1,2)]*\\
\frac{(\sigma_1(1,2)+\sigma_2(1,2)+\sigma_3(1,2))*(E[\hat{\theta}(1,1)]+E[\hat{\theta}(2,2)]+E[\hat{\theta}(1,3)])}{3}\\
         \quad\quad     \vdots   \quad\quad\quad\quad \quad\quad\quad\quad \vdots   \quad\quad\quad\quad \quad\quad\quad\quad \vdots\\

\forall a\in (1,A), \forall b\in(1,B) \\ 
E(\hat{\theta}(a,b))=Pr(a,b)*\theta(a,b)+ [1-Pr(a,b)]*\\
   
   \frac{(\sigma_1(a,b)+\dots+\sigma_4(a,b))*(E[\hat{\theta}(a-1,b)]+\dots+E[\hat{\theta}(a,b+1)])}{4}\\
         \quad\quad     \vdots   \quad\quad\quad\quad \quad\quad\quad\quad \vdots   \quad\quad\quad\quad \quad\quad\quad\quad \vdots\\
         
E(\hat{\theta}(A,B))=Pr(A,B)*\theta(A,B)+ [1-Pr(A,B)]*\\
   \frac{(\sigma_1(A,B)+\sigma_2(A,B))*(E[\hat{\theta}(A-1,b)]+E[\hat{\theta}(A,B-1)]}{2}\\
             \end{array}
\right.
\end{equation*}

Note that the above equations can be simply extended to all the pixels in $\Upsilon_j$ in all the frames (incorporating the frame ID $t$). We use the inverse matrix to solve these equations where the coefficients of all the above equations can be represented as a $|\Upsilon_j|\times |\Upsilon_j|$ matrix (denoted as $M$). To ensure that the inverse matrix can solve the equations, $M$ should have a full rank $|\Upsilon_j|$. In case that $M$ is not a full rank matrix (indeed, the rank of $M$ is very high since $\forall (a, b, t)\in \Upsilon_j, \sigma_1(a,b,t), \sigma_2$ $(a,b,t), \sigma_3(a,b,t), \sigma_4(a,b,t)$ are pseudorandom), we can add a tiny random noise to the non-zero entries in $M$ (in which the deviation is negligible).

 Specifically, denoting the expectation of the $s$th pixel in $\Upsilon_j$ as $E[\hat{\theta}(s)]$ where $s\in[1,AB]$. Then, we have
 
 \begin{equation}
 \small
E[\hat{\theta}(s)]= \frac{1}{|M|}*\sum_{i=1}^{AB}[(-1)^{i+s}*M_{is}^{(AB-1)}*b_i]
 \end{equation}

where $|M|$ is the determinant of matrix $M$, $M_{is}^{(AB-1)}$ denotes the $s$th cofactor (corresponding the $s$th pixel; including $(AB-1)\times(AB-1)$ entries) and $b_i$ is the $i$th constant in the equation (in last column of matrix $M$). Thus, $M_{is}^{(AB-1)}$ can be recursively represented as below:

\begin{equation}
\small
M_{is}^{(AB-1)}=\sum_{i=1}^{AB}[(-1)^{i+s}*\mathbb{R}_i*M_{is}^{(AB-2)}]
\label{eq:mis}
\end{equation}
 
where $M^{(AB-2)}$ represents the cofactor matrix of $M^{AB-1}$ and $\mathbb{R}_i$ is a random constant (for ensuring full rank for $M$) which is close to $-\frac{[1-Pr(a,b,t)](\sigma_1(a,b,t)+\sigma_2(a,b,t))}{2}$ for corner pixels, $\\-\frac{[1-Pr(a,b,t)][\sigma_1(a,b,t)+\sigma_2(a,b,t)+\sigma_3(a,b,t)]}{3}$ for border pixels, and $-\frac{[1-Pr(a,b,t)][\sigma_1(a,b,t)+\sigma_2(a,b,t)+\sigma_3(a,b,t)+\sigma_4(a,b,t)]}{4}$ for non-border pixels. As a result, Equation \ref{eq:mis} can be represented as:

 \begin{equation}
\small
M_{is}^{(AB-1)}=\sum_{i= 1}^{AB}[(-1)^{i+s}*(\prod_{i=1}^{AB-3}\mathbb{R}_i)*M_{is}^{(2)}]
 \end{equation}
 
Since each row of the matrix $M$ only has at most 5 non-zero entries (corresponding to the variables of the current pixel and its four/three/two neighbors), we have:

\begin{equation}
\small
E[\hat{\theta}(s)]\approx -\frac{5^{AB-3}*AB}{|M|}*\max_{ \forall i\in[1,AB]}\{|\mathbb{R}_i|^{AB-3}*M_{is}^{(2)}*b_i\}
\end{equation}
 
Thus, the MSE expectation in VE $\Upsilon_j$ can be directly derived as:
 
\begin{equation*}
\small
\sum_{i=1}^{AB}[\theta(a,b,t)+\frac{5^{AB-3}*AB}{|M|}*\max_{ \forall i\in[1,AB]}\{|\mathbb{R}_i|^{AB-3}*M_{is}^{(2)}*b_i\}]^2
\end{equation*}

For each $k_j$, the corresponding MSE expectation can be computed using the above equation. Then, the optimal $k_j$ can be obtained by traversing $k_j$ in any range. In addition, it is straightforward to prove that the complexity of the inverse matrix based solver is $O(n^3\log(n))$. Note that we assume that the optimal $k_j$ is computed for minimum MSE based on the first traversal in the interpolation of each visual element (in Algorithm \ref{algm:inter}). The deviation is very minor since most pixels are interpolated in the first traversal in our experiments. Moreover, the optimal $k_j$ (derived from the above algorithm) is also validated in our experiments (see Figure \ref{fig:msekafter}).

\section{Budget Allocation Algorithm}
\label{sec:baa}

\begin{algorithm}[!h]
\small
\SetKwInOut{Input}{Input}\SetKwInOut{Output}{Output}

\Input{$n$ sets of RGBs $\forall j\in[1,n], \Psi_j=\{i\in [1,k_j], \widetilde{\theta}_{ij}\}$\\privacy budget $\epsilon$ for Phase I of \emph{VideoDP}}

\Output{privacy budget for each unique RGB in $n$ sets}

initialize the set of unique RGBs: $\Psi\leftarrow\bigcup_{j=1}^n\Psi_j$

\ForEach{$j\in[1,n]$}{
initialize the overall budget for set $\Psi_j: \epsilon(\Psi_j)\leftarrow\epsilon$
}

\ForEach{$\ell\in[1,n]$}
{
    \ForEach{$\widetilde{\theta}\in \Psi$}{
    \If{$count(\widetilde{\theta}\in \{\Psi_1,\dots, \Psi_n\})=(n-\ell+1)$}{
    
    \tcp{w.l.o.g., $\widetilde{\theta}\in\Psi_1, \dots, \Psi_{n-\ell+1}$}
    
    initialize budget for RGB $\widetilde{\theta}$: $\epsilon(\widetilde{\theta})$

    $\displaystyle\epsilon(\widetilde{\theta})\leftarrow\argmin_{\forall j\in [1,(n-\ell+1)]} [\frac{d_j(\widetilde{\theta})}{d_j}*\epsilon(\Psi_j)]$
    
    \tcp{$\frac{d_j(\widetilde{\theta})}{d_j}$ denotes the ratio of pixels with RGB $\widetilde{\theta}$ in $\Psi_j$}
    
    \ForEach{$j\in[1,(n-\ell+1)], \Psi_j$}{
        update budget: $\epsilon(\Psi_j)\leftarrow \epsilon(\Psi_j)-\epsilon(\widetilde{\theta})$
    
        update total pixel count: $d_j\leftarrow d_j-d_j(\widetilde{\theta})$
        
        \tcp{$\epsilon(\Psi_j)>\epsilon(\widetilde{\theta})$ for the first $(n-1)$ partitions; $\epsilon(\Psi_j)=\epsilon(\widetilde{\theta})$ for the last partition (0 budget left)}
    }
    
    \textbf{return} budget $\epsilon(\widetilde{\theta})$ for RGB $\widetilde{\theta}$
    
    $\Psi\leftarrow \Psi\setminus\widetilde{\theta}$

    }
    
	}

}
	\caption{Budget Allocation}
	\label{algm:budget}
\end{algorithm}    

\section{Probabilistic Differential Privacy and Indistinguishability}
\label{sec:indist}

\begin{proposition}
If for any two neighboring inputs $V$ and $V'$, $\frac{Pr[\mathcal{A}(V)=O]}{Pr[\mathcal{A}(V')= O]}\leq
e^\epsilon$ hold (where $O$ is an arbitrary output), 
then $\frac{Pr[\mathcal{A}(V)\in
S]}{Pr[\mathcal{A}(V')\in S]}\leq e^\epsilon$ also holds (where $S$ is an arbitrary set of outputs).
\label{thm:pdp}
\end{proposition}

\begin{proof}
Since $S$ includes a set of possible outputs, we have:

\begin{align*}
&Pr[\mathcal{A}(V)\in S]= \int_{\forall O\in
S}Pr[\mathcal{A}(V)=O]dO&\\ \leq &e^\epsilon\int_{\forall
O\in S}Pr[\mathcal{A}(V')=O]dO=e^\epsilon Pr[\mathcal{A}(V')\in S]&
\end{align*}

This completes the proof. Note that the above proof is adapted from \cite{MachanavajjhalaKAGV08,corn09}. 
\end{proof}

\section{Additional Evaluations}
\label{sec:additional}

While evaluating the utility of the utility-driven private videos using two utility measures (KL divergence and MSE), we also fix $\epsilon$ and traverse different $k$ for all the visual elements (assigning the same $k\in [4,30]$). Figure \ref{fig:klk1} and \ref{fig:klk2} present the KL divergence values for all the sampled pixels (where privacy budget $\epsilon$ is fixed as $0.8$ and $1.6$, respectively). We can observe that the KL value increases as $k$ increases (if the same number of distinct RGBs in all the visual elements are selected to assign privacy budgets). This is true for the following reason: smaller $k$ samples pixels with less diverse RGBs, but it can allocate a larger privacy budget to each RGB. Then, the generated results can have better count distributions for all the sampled RGBs.

\begin{figure}[!h]
	\centering
		\subfigure[KL vs $k$ (Video PED)]{
		\includegraphics[angle=0, width=0.48\linewidth]{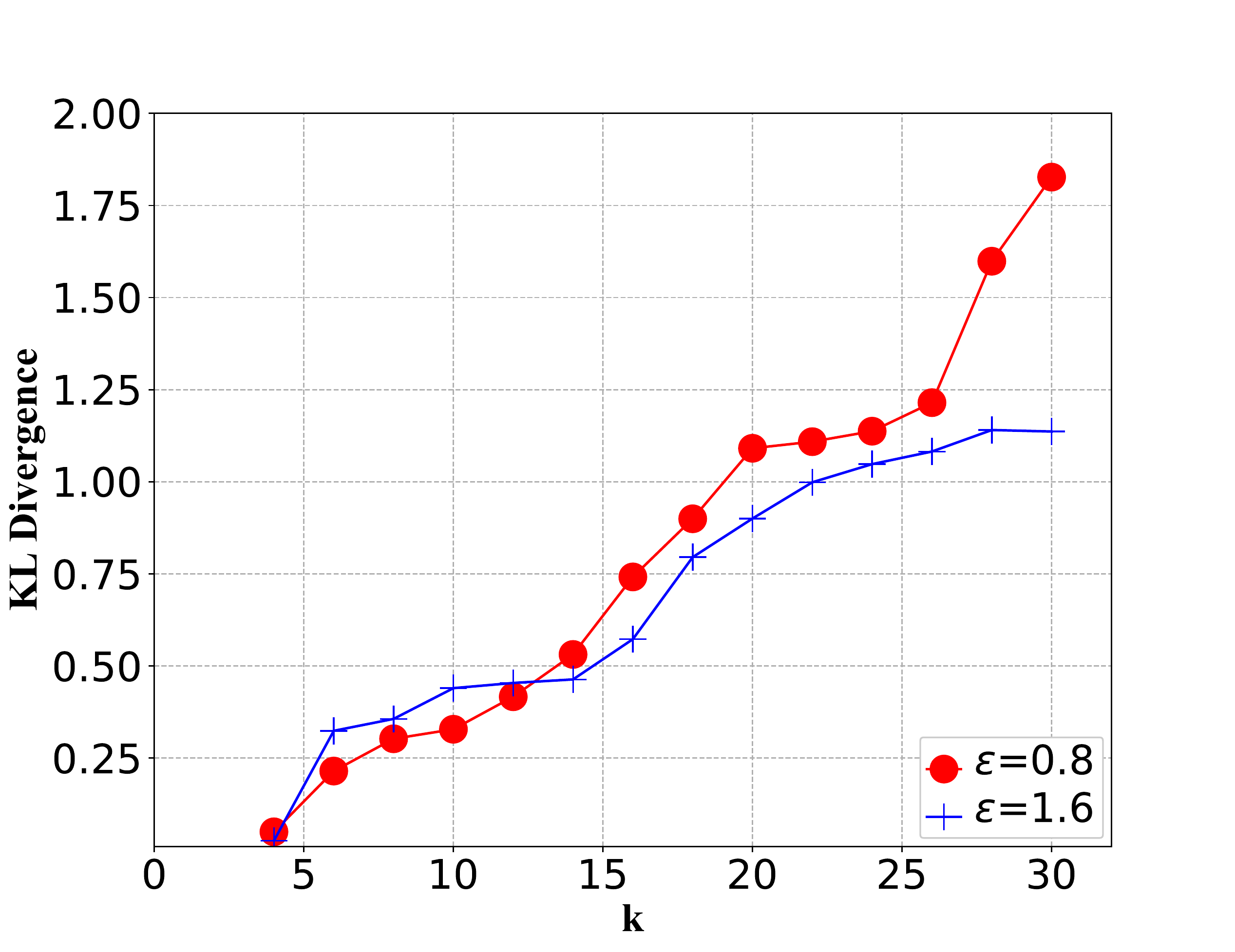}
		\label{fig:klk1}}
		\hspace{-0.17in}
	\subfigure[KL vs $k$ (Video VEH)]{
		\includegraphics[angle=0, width=0.48\linewidth]{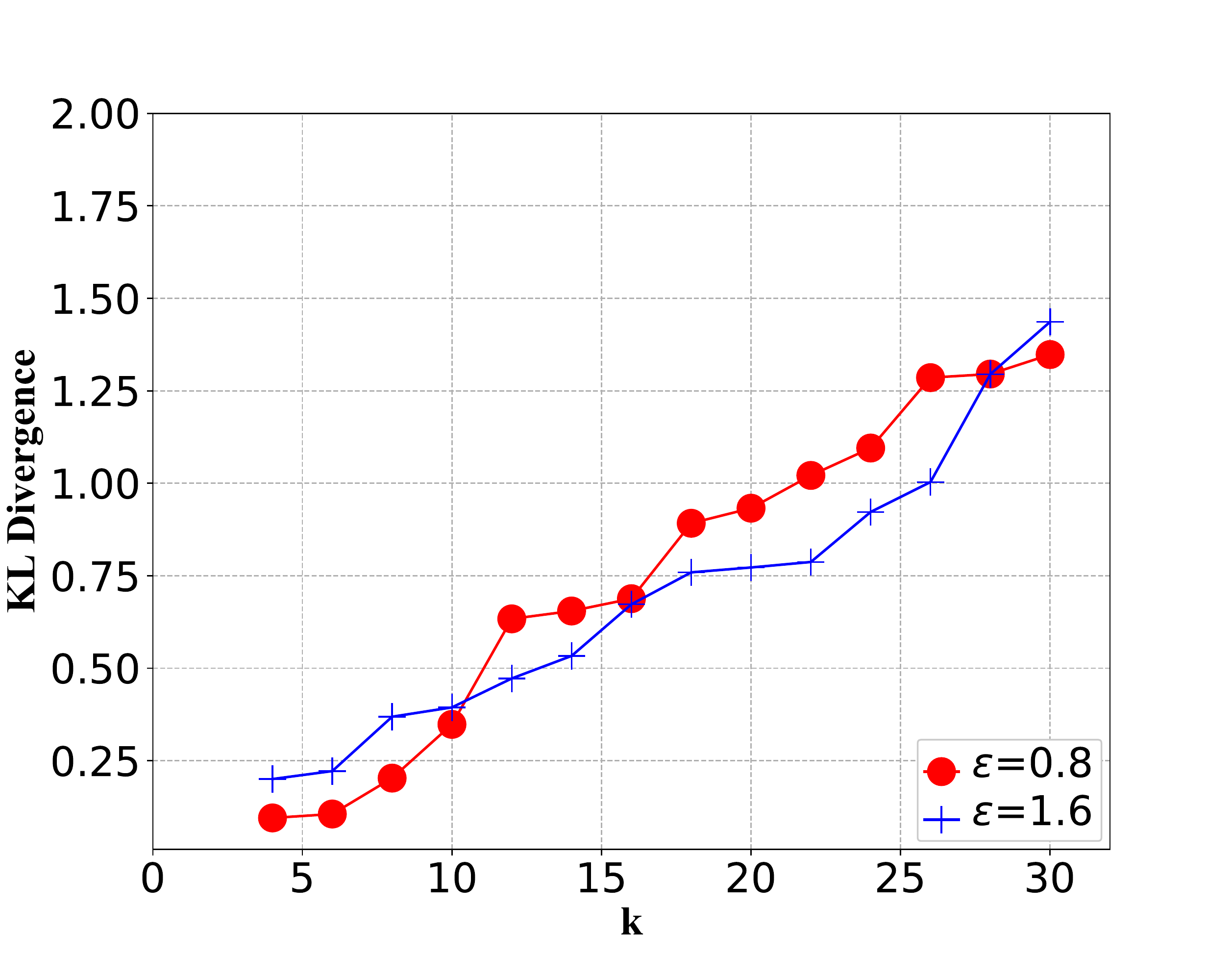}
		\label{fig:klk2}}
		\hspace{-0.2in}
		\subfigure[MSE vs $k_j$ (after Phase I)]{
		\includegraphics[angle=0, width=0.48\linewidth]{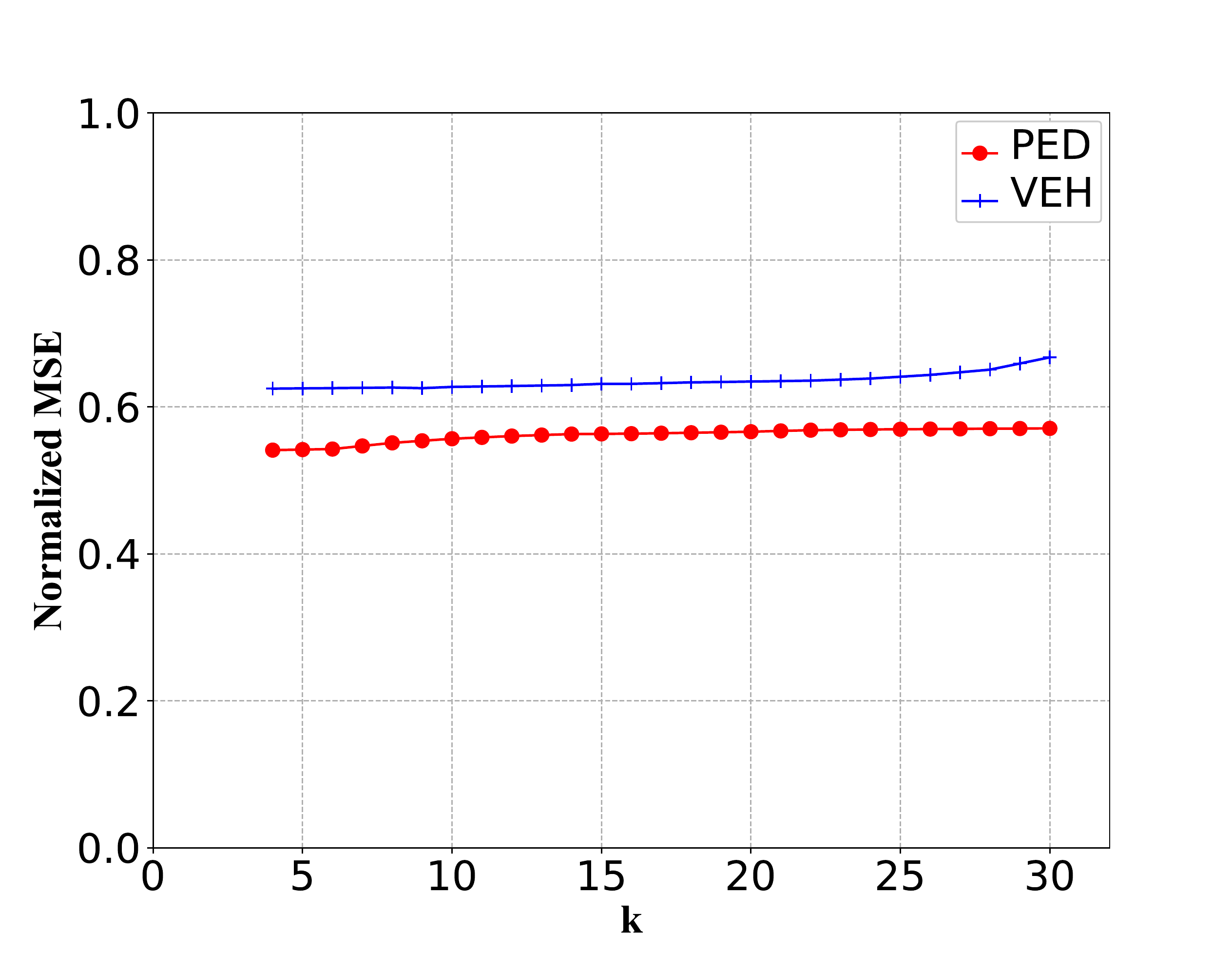}
		\label{fig:msekbefore} }
		\hspace{-0.17in}
	\subfigure[MSE vs $k_j$ (after Phase II)]{
		\includegraphics[angle=0, width=0.48\linewidth]{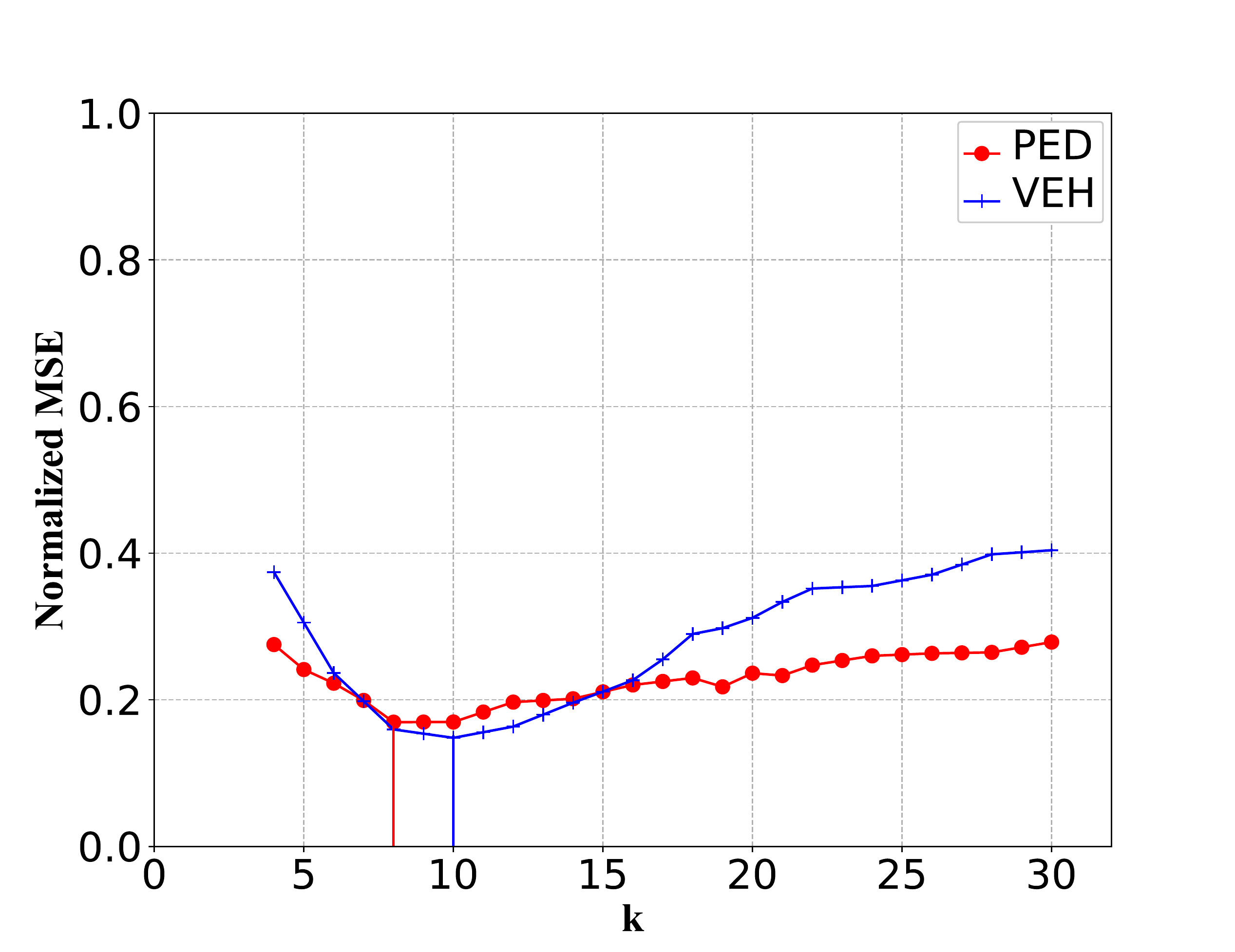}
		\label{fig:msekafter}}
	
	\vspace{-0.05in}
	\caption[Optional caption for list of figures]
	{Pixel Level Utility Evaluation with $k$}\vspace{-0.1in}
	\label{fig:mse}
\end{figure}

We also examine the \emph{optimal number of selected RGBs to assign privacy budgets} $k_j$ in visual elements. We select the visual element with most pixels in both videos (PED and VEH). Since the optimal values are derived based on MSEs, we plot the normalized MSEs for \emph{all the pixels in the visual element} for two videos in Figure \ref{fig:msekbefore} (after Phase I) and Figure \ref{fig:msekafter} (after Phase II), respectively. The normalized MSE does not change much (after Phase I) as $k$ increases since the MSE expectation is optimized for Phase II. Instead, Figure \ref{fig:msekafter} clearly shows that $k_j$ goes \emph{optimal} in the range (which equals the optimal $k_j$ after solving Equation \ref{eq:mse} using the algorithm in Appendix \ref{sec:opt}) in both videos for all possible values in the specified range. As $k_j$ increases, the normalized MSE of the VE first decreases and then increases. This reflects that the best $k_j$ is neither too small nor too large in each VE.

Finally, we present some representative frames of the VEH video to show the effectiveness of pixel sampling (Phase I) and utility-driven private video generation (Phase II) in \emph{VideoDP}. Figure \ref{fig:VEH_sam1} and \ref{fig:VEH_sam2} demonstrate that more pixels are sampled as private budget $\epsilon$ is larger. The same frames (missing pixels are interpolated) after Phase II are shown in Figure \ref{fig:VEH_ep1} and \ref{fig:VEH_ep2}, respectively. We can observe that the vehicles are randomly generated in the frame of the utility-driven private video (which are not directly revealed to the analysts). Then, disclosing the any query/analysis result derived from such (random) video to untrusted analysts satisfies differential privacy.
\begin{figure}[!tbh]
	\centering
	    \subfigure[$\epsilon=0.8$ (after Phase I)]{
		\includegraphics[angle=0, width=0.45\linewidth]{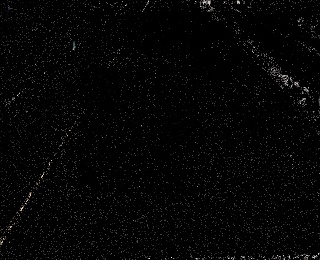}
		\label{fig:VEH_sam1} }
        \subfigure[$\epsilon=1.6$ (after Phase I)]{
		\includegraphics[angle=0, width=0.45\linewidth]{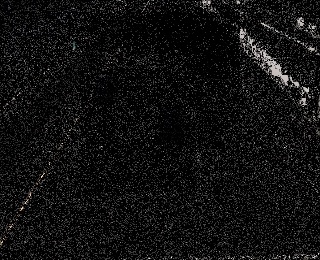}
		\label{fig:VEH_sam2}}
		\subfigure[$\epsilon=0.8$ (after Phase II)]{
		\includegraphics[angle=0, width=0.45\linewidth]{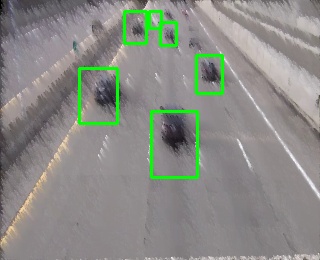}
		\label{fig:VEH_ep1} }
        \subfigure[$\epsilon=1.6$ (after Phase II)]{
		\includegraphics[angle=0, width=0.45\linewidth]{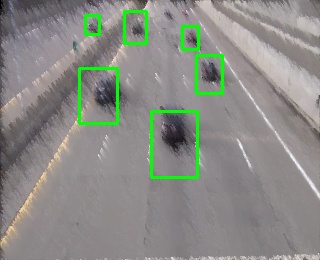}
		\label{fig:VEH_ep2} }
		\vspace{-0.1in}
	\caption[Optional caption for list of figures]
	{Representative Frames of \emph{VideoDP} on Video VEH}
	\label{fig:inVEH}
\end{figure}

\section{The Notation Table}
\label{sec:notation}

\begin{small}
\begin{table}[!h]
\small
\centering
\caption{Frequently Used Notations}\vspace{-0.1in}
\begin{tabular}{|c|l|} 
\hline
VE & visual element (e.g., object, human)\\

$V,O$ & orignal video and output synthetic video \\

$|V|,|O|$ & total pixel counts in $V$, $O$\\

$m$ & the number of distinct RGBs in $V$\\

$\theta_i$ & the $i$th RGB in $V$ where $i\in[1,m]$\\

$n$ & the number of distinct VEs in $V$\\

$\Upsilon_j$ & the $j$th VE in $V$ (all the frames), $j\in[1,n]$\\

$|\Upsilon_j|$ & total number of pixels in $|\Upsilon_j|$\\

$\Psi_j$ & set of RGBs in $\Upsilon_j$ with budgets\\

$|\Psi_j|$ & cardinality of $\Psi_j$\\

$d_j$ & total pixel count in $\Upsilon_j$\\

$\widetilde{\theta}_{ij}$  & the $i$th RGB in $\Psi_j$\\

$k_j$ & (optimal) number of distinct RGBs in $\Upsilon_j$ \\

$\Psi$, $|\Psi|$ & $\bigcup_{j=1}^n\Psi_j$, cardinality of $\Psi$\\

$\widetilde{\theta}_i, \theta_i$ & the $i$th RGB in $\Psi$, the $i$th RGB in $V$\\

$\widetilde{c}_i$ (or $c_i$), $\widetilde{c}_i^j$ & total pixel count for RGB $\widetilde{\theta}_i$ (or $\theta_i$) in $V$, $\Upsilon_j$\\

$\widetilde{x}_i$ (or $x_i$) & total pixel count for RGB $\widetilde{\theta}_i$ (or $\theta_i$) in $O$\\

$(a,b,t)$ & the pixel with coordinates $(a,b)$ and frame $t$\\

$\theta(a,b,t)$ & the RGB of pixel $(a,b,t)$ in $V$\\

$\hat{\theta}(a,b,t)$ & the RGB of pixel $(a,b,t)$ in $O$\\

$Pr(a,b,t)$ & probability that pixel $(a,b,t)$ is sampled\\

$\sigma_0,\dots, \sigma_4$ &probabilities that pixel $(a,b,t)$ has $0, 1, \dots, 4$ \\
& neighboring pixels after Phase I (sampling)\\

$\hat{\theta}_N$& simplified notation for $\hat{\theta} (a-1,b,t)$\\

$\hat{\theta}_S$& simplified notation for $\hat{\theta} (a+1,b,t)$\\

$\hat{\theta}_W$& simplified notation for $\hat{\theta} (a,b-1,t)$\\

$\hat{\theta}_E$& simplified notation for $\hat{\theta} (a,b+1,t)$\\

\hline
\end{tabular}

\label{table:notation}
\end{table}
\end{small}

\end{appendices}

\balance

\end{document}